\documentclass[12pt, draftclsnofoot, onecolumn]{IEEEtran}

\usepackage{url}
\usepackage{xcolor}
\usepackage{cite,graphicx,amsmath,amssymb}
\usepackage[outdir=./]{epstopdf}
\usepackage{subfigure}
\usepackage{fancyhdr}
\usepackage{mdwmath}
\usepackage{mdwtab}
\usepackage{caption}
\usepackage{amsthm}

\usepackage{algorithmic,algorithm}%

\newtheorem{definition}{Definition}

\newtheorem{theorem}{Theorem}

\newtheorem{lemma}{Lemma}

\newtheorem{corollary}{Corollary}

\newtheorem{proposition}{Proposition}

\allowdisplaybreaks[4]
\makeatletter
\def\ScaleIfNeeded{%
\ifdim\Gin@nat@width>\linewidth \linewidth \else \Gin@nat@width
\fi } \makeatother

\begin{document}

\title{ Multi-Agent Reinforcement Learning Based Resource Allocation  for  UAV  Networks}
\vspace{-1.9em}
\author{
Jingjing~Cui, ~\IEEEmembership{Member,~IEEE,}
Yuanwei~Liu,~\IEEEmembership{Member,~IEEE,}
Arumugam~Nallanathan,~\IEEEmembership{Fellow,~IEEE,}

\thanks{J.cui, Y. Liu and A. Nallanathan are with the School of Electronic Engineering and Computer Science, Queen Mary University of London, London E1 4NS, U.K. (email: \{j.cui, yuanwei.liu, a.nallanathan\}@qmul.ac.uk).}

\thanks{}  }

\maketitle
\vspace{-2.9em}
\begin{abstract}
Unmanned aerial vehicles (UAVs) are capable of serving as aerial base stations (BSs) for providing both  cost-effective and on-demand wireless communications. This article investigates  dynamic resource allocation  of multiple UAVs   enabled communication networks  with the goal of maximizing  long-term rewards. More particularly,  each UAV   communicates with a ground user by automatically selecting its communicating  users, power levels and  subchannels without any information exchange among  UAVs. 
To model the uncertainty of environments, we formulate the  long-term resource allocation problem  as a stochastic game for maximizing the expected rewards, where each UAV becomes a learning agent and each resource allocation solution corresponds to an action taken by  the UAVs.  Afterwards, we develop a multi-agent reinforcement learning (MARL) framework that each agent  discovers its best strategy according to its local observations using learning. More specifically,  we propose  an agent-independent method,  for which all agents conduct a decision algorithm independently but share a common structure based on Q-learning.  
Finally, simulation results reveal that: 1) appropriate parameters for  exploitation and exploration are capable of enhancing the performance of the proposed MARL based resource allocation algorithm; 2) the proposed MARL algorithm provides acceptable performance compared to the case with complete information exchanges among UAVs. By doing so,  it   strikes  a good tradeoff  between performance gains and  information exchange overheads.

\end{abstract}

\section{Introduction}

Aerial communication networks, encouraging new innovative functions to deploy wireless infrastructure, have recently  attracted increasing interests  for providing high network capacity and enhancing coverage  \cite{Bucaille13Milcom}.  Unmanned aerial vehicles (UAVs), also known  as remotely piloted aircraft systems (RPAS) or drones, are small pilotless
aircraft that  are  rapidly deployable for complementing  terrestrial communications based on  the 3rd Generation Partnership Project (3GPP) LTE-A (Long term evolution-advanced)~\cite{Chandrasekharan16Mag}. In contrast to  channel characteristics of  terrestrial communications, the channels of UAV-to-ground communications   
are more probably line-of-sight (LoS) links \cite{Hourani14WCL}, which is beneficial for wireless communications.   

In particular, UAVs based different aerial platforms that for providing wireless services have  attracted  extensive research and industry efforts  in terms of the issues of deployment, navigation and control \cite{Mozaffari18Tutorial}.   Nevertheless, resource allocation such as transmit power, serving users and subchannels, as a key communication problem, is also essential to further enhance the energy-efficiency and coverage for UAV-enabled communication networks.  

\subsection{Prior Works}

Compared to terrestrial BSs, UAVs are generally faster  to deploy and  more flexible to configure. The deployment of UAVs in terms of  altitude and distance between UAVs  was investigated for UAV-enabled small cells in \cite{Mozaffari15aDBLP}.  In  \cite{Mozaffari16CL}, a three-dimensional (3D) deployment algorithm based on circle packing is developed for maximizing the downlink coverage performance. Additionaly, a 3D deployment algorithm of a single UAV is developed for maximizing the number of covered users in  \cite{Alzenad17WCL}.  Moreover, by fixing the altitudes, a successive UAV placement approach was proposed to minimize the number of UAVs required while guaranteeing each ground user to be covered by at least one UAV in \cite{Lyu17CL}.  
 
Despite the deployment optimization of UAVs,  trajectory designs of UAVs for optimizing the communication performance have  attracted tremendous  attentions, such as in \cite{Zeng16TC,Zeng18TWC, Wu18TWC}. In \cite{Zeng16TC}, the authors considered one UAV as a mobile relay and investigated the throughput maximization problem by optimizing power allocation and the UAV's trajectory. Then, a designing approach of the UAV's trajectory based on successive convex approximation (SCA) techniques was  proposed in \cite{Zeng16TC}. By transforming the continuous trajectory into a set of discrete waypoints, the authors in \cite{Zeng18TWC} investigated  the UAV's trajectory design with minimizing the mission completion time in a  UAV-enabled multicasting system. Additionally, multiple-UAV  enabled wireless communication networks (multi-UAV networks) were considered in \cite{Wu18TWC}, where a joint design for optimizing trajectory and  resource allocation was studied with the goal of guaranteeing fairness by maximizing the minimum throughput among users.  In \cite{Zhang2018arXiv}, the authors proposed a joint of subchannel assignment and trajectory  design approach  to  strike a tardeoff between the sum rate and the delay of sensing tasks  for a multi-UAV aided uplink single cell network.

Due to the versatility and manoeuvrability of UAVs, human intervention becomes restricted for UAVs' control design. Therefore,  machine learning based intelligent control of UAVs is desired for enhancing the performance for UAV-enabled communication networks. Neuro network based trajectory designs were considered from the perspective of UAVs' manufactured structures    in \cite{geiger2009neural} and \cite{Nodland11Cof}. Regarding UAVs enabled communication networks, a weighted expectation based predictive on-demand deployment approach of UAVs was proposed to minimize the transmit power  in \cite{ZhangQ2018Glb}, where  Gaussian mixture model was used for building  data distributions.  In \cite{Challita18Janarxiv}, the authors studied the autonomous path planning  of UAVs by jointly taking energy efficiency, lantency and interference into consideration, in which a echo state network based deep reinforement learning algorithm was proposed.     In \cite{Chen17Glb}, the authors proposed a liquid state machine (LSM) based resource allocation algorithm for cache enabled UAVs  over LTE licensed and unlicensed bands. Additionally, a log-linear learning based joint channel-slot selection algorithm was developed for multi-UAV networks in \cite{Chen18Ac}.

\subsection{Motivation and Contributions}

As discussed above, machine learning is a promising and power tool to provide    autonomous and effective solutions in an intelligent manner to enhance the UAV-enabled communication networks. However, most research contributions  focus on the  deployment and trajectory designs of UAVs in communication networks, such as \cite{ZhangQ2018Glb,Challita18Janarxiv,Chen17Glb}. Though resource allocation schemes such as transmit power and subchannels were considered  for UAV-enabled communication networks in \cite{Wu18TWC} and \cite{Zhang2018arXiv}, the prior studies focused on time-independent scenarios.  That is the optimization design is independent for each time slot. Moreover,  for time-dependent scenarios,  \cite{Chen17Glb} and \cite{Chen18Ac} investigated the potentials of machine learning based  resource allocation algorithms. However, most of the proposed machine learning algorithms mainly focused on single UAV scenarios or multi-UAV scenarios by assuming the availability of complete network information for each UAV.   In practice, it is non-trivial  to obtain perfect knowledge of dynamic environments due to the high movement speed of UAVs \cite{Sun13Milcom,Cai13TVT}, which imposes formidable challenges on the design of reliable UAV-enabled wireless communications. Besides, most  existing research contributions focus on centralized approaches, which makes  modeling and computational tasks become challenging as the network size continues to increase. Multi-agent reinforcement learning (MARL)  is capable of providing a distributed perspective on the intelligent resource management for UAV-enabled communication networks especially   when  these UAVs only have individual local information.

The main benefits of MARL  are: 1) agents consider individual application-specific nature and environment; 2) local exchanges between agents can be modeled and investigated; 3) difficulties in modelling and computation can be handled in distributed manners.  The applications of  MARL  for  cognitive radio networks were studied in  \cite{Li09Cof} and \cite{Serrano10TVT}. Specifically, in  \cite{Li09Cof},  the authors focused on the feasibilities of MARL based  channel selection algorithms for a specific  scenario with two secondary users. A real-time aggregated interference scheme based on MARL was investigated in \cite{Serrano10TVT} for wireless regional area networks (WRANs). Moreover, in \cite{Asheralieva16TC}, the authors proposed a MARL based channel and power level selection algorithm  for device-to-device (D2D) pairs in heterogeneous cellular networks. The potential of  machine learning based user clustering for mmWave-NOMA networks  was presented in \cite{Cui18NOMAML}. Therefore, invoking MARL to UAV-enabled communication networks provides a promising solution for intelligent resource management.  
Due to the high mobility and adaptive altitude, to the best of our knowledge, multi-UAV networks are not well-investigated, especially for the resource allocation from the perspective of MARL. However, it is challenging  for MARL based multi-UAV networks to specify   a suitable objective  and strike a exploration-exploitation tradeoff.  

Motivated by the features of MARL and  UAVs, this article aims to develop a MARL framework for multi-UAV networks. More specifically, we consider a multi-UAV enabled downlink wireless network, in which multiple UAVs try to communicate with ground users simultaneously. Each UAV  flies according to the predefined trajectory. It is assumed  that  all UAVs communicate with ground users without the assistance of a central controller.  Hence,  each UAV can only observe its local information. Based on the proposed framework, our major contributions are summarized as follows:
\begin{enumerate}

\item  We investigate the optimization problem of maximizing long-term rewards of multi-UAV downlink networks by jointly designing user, power level and subchannel selection strategies.  Specifically,  we formulate a quality of service (QoS) constrained  energy efficiency function as the reward function for providing a reliable communication. Because of the time-dependent  nature and environment uncertainties, the formulated optimization problem  is non-trivial.  To solve the challenging problem, we propose a learning based dynamic resource allocation algorithm.

\item We propose  a novel framework based on stochastic game  theory \cite{nowe2012game}  to model the dynamic resource allocation problem of   multi-UAV networks, in which  each UAV becomes a learning agent and each resource allocation solution corresponds to an action taken by  the UAVs. Particularly, in the formulated stochastic game, the actions for each UAV satisfy the properties of Markov chain \cite{neyman2003markov}, that is the reward of a UAV is only dependant on the current state and action. Furthermore, this framework can  be also applied to model the resource allocation problem for a wide range of dynamic multi-UAV systems.

\item We develop a MARL based resource allocation algorithm for solving the formulated stochastic game  of multi-UAV networks. Specifically,  
each UAV as an independent learning agent runs a standard Q-learning algorithm by ignoring the other UAVs, and hence information exchanges   between UAVs and computational burdens on each UAV are substantially reduced.  Additionally, we also provide a convergence proof of the proposed MARL based resource allocation algorithm.

\item Simulation results are provided to derive parameters for  exploitation and exploration  in the $\epsilon$-greedy method over different network setups. Moreover, simulation results also demonstrate that the proposed MARL based resource allocation framework for multi-UAV networks strikes a good tradeoff  between performance gains and  information exchange overheads.

\end{enumerate}

\subsection{Organization}
The rest of this article is organized as follows. In Section II, the system model for downlink multi-UAV  networks  is presented. The problem of resource allocation is formulated and a stochastic game framework for the considered  multi-UAV  network  is presented in Section III. In Section IV, a Q-learning based MARL algorithm for resource allocation is designed. Simulation results are presented in Section V, which is followed by the conclusions in Section VI.

\section{System Model}
\begin{figure} [t!]
\centering
\includegraphics[width= 3.5in, height=2.6in]{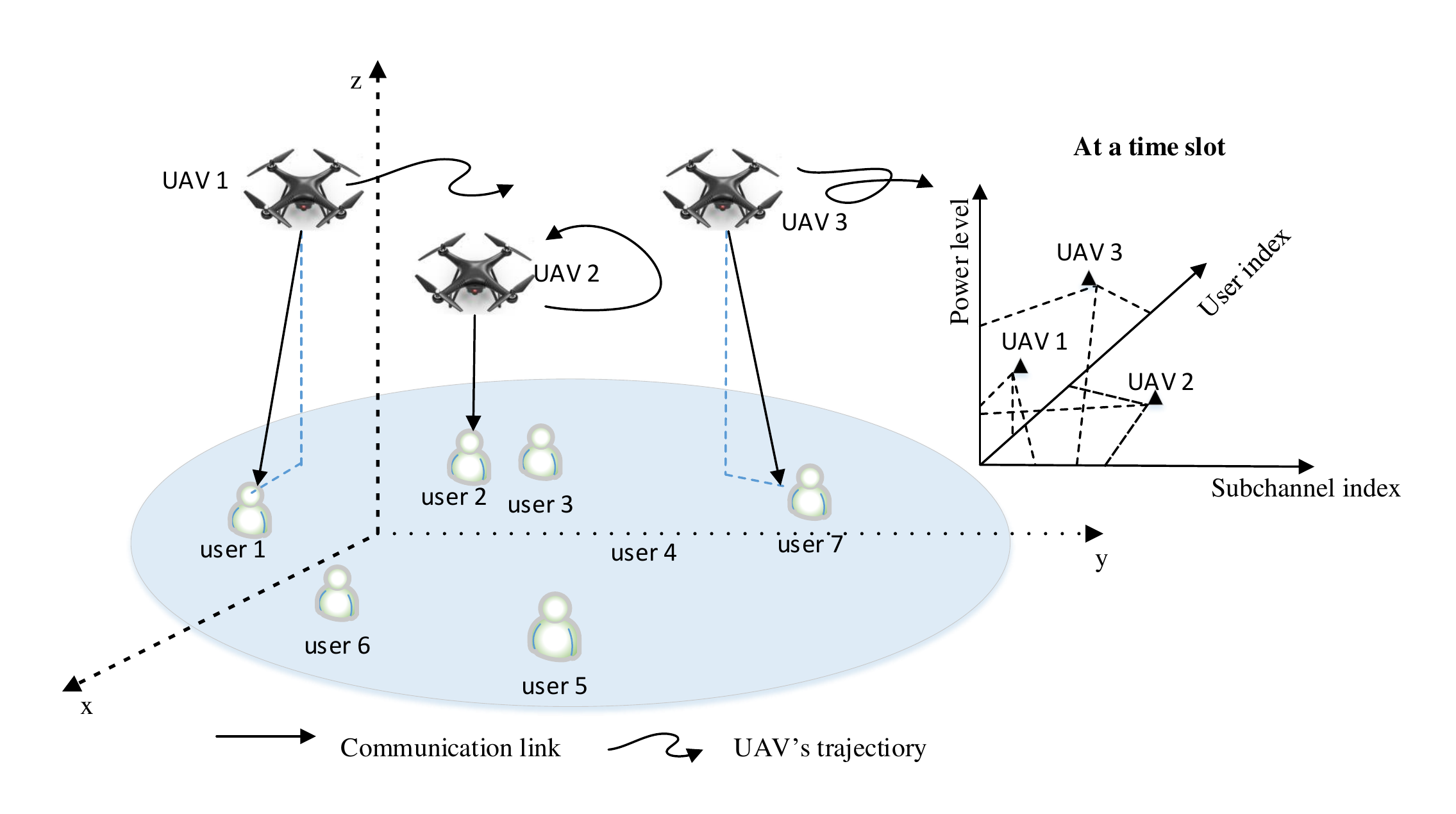}
  \vspace*{-0.5em}\caption{  Illustration of multi-UAV communication networks.}
 \label{UAV_system_model}
   \vspace{-1.0em}
\end{figure}
Consider a multi-UAV downlink communication network as illustrated in Fig. \ref{UAV_system_model} operating in a discrete-time axis, which consists of  $M$ single-antenna UAVs and $L$  single-antenna users, denoted by  $\mathcal{M}=\{1,\cdots, M\}$ and $\mathcal{L}=\{1,\cdots,L\}$, respectively. The ground users are randomly distributed in the considered disk with radius $r_d$. As shown in  Fig. \ref{UAV_system_model}, multiple UAVs fly over this region and  communicate with ground users by providing direct communication connectivity  from the sky \cite{Bucaille13Milcom}.  The total bandwidth $W$ that the UAVs can operate  is divided into $K$ orthogonal subchannels, denoted by $\mathcal{K}=\{1,\cdots,K\}$. Note that the subchannels occupied by UAVs may  overlap with each other.  Moreover, it is assumed that UAVs fly autonomously without human intervention based on pre-programmed flight plans as in \cite{how2004flight}. That is the trajectories of UAVs are predefined based on the pre-programmed flight plans.  As shown in Fig. \ref{UAV_system_model}, there are three UAVs flying on the considered region based on the pre-defined trajectories, respectively. This article focuses on the dynamic design of  resource allocation for multi-UAV networks in term of user, power level and  subchannel selections.   
Additionally, assuming that all UAVs communicate without the  assistance of a central controller and  have no global knowledge of  wireless communication environments. In other words,  the channel state information (CSI) between a   UAV and users are known locally. This assumption is reasonable in practical due to the mobilities of UAVs, which is similar to the research contributions  \cite{Sun13Milcom,Cai13TVT}.

\subsection{UAV-to-Ground Channel Model}

In contrast to the propagation of   terrestrial  communications,   the air-to-ground (A2G) channel is highly dependent on the altitude, elevation angle and the type of the propagation environment  \cite{Hourani14WCL,Mozaffari18Tutorial,Chandrasekharan16Mag}. 
 In this article, we investigate the dynamic resource allocation problem for    multi-UAV networks under two types of UAV-to-ground  channel models:

\subsubsection{Probabilistic Model}
 As discussed in \cite{Hourani14WCL,Chandrasekharan16Mag}, UAV-to-ground communication links can be modeled by a probabilistic path loss model, in which the LoS and non-LoS (NLoS) links can be considered separately with different probabilities of occurrences. According to \cite{Hourani14WCL}, at time slot $t$, the probability of having a LoS connection between UAV $m$ and a ground user $l$ is given by
\begin{align}
P^{\mathrm{LoS}}(t) = \frac{1}{1 + \mathfrak{a}\exp(-\mathfrak{b}\sin^{-1}(\frac{H}{d_{m,l}(t)})- \mathfrak{a})},
\end{align}
where $\mathfrak{a}$ and $\mathfrak{b}$ are constants that depend on the environment. $d_{m,l}$ denotes the distance between UAV $m$ and user $l$ and $H$ denotes the altitude of UAV $m$. Furthermore, the probability of have NLoS links is $P^{\mathrm{NLoS}}(t) = 1 - P^{\mathrm{LoS}}(t)$. 

Accordingly, in time slot $t$, the LoS and NLoS pathloss from UAV $m$ to the ground user $l$ can be expressed as
\begin{subequations}
\begin{align}
PL_{m,l}^{\mathrm{LoS}} &= L^{\mathrm{FS}}_{m,l}(t)+  \eta^{\mathrm{LoS}},\\
PL_{m,l}^{\mathrm{NLoS}} &= L^{\mathrm{FS}}_{m,l}(t) + \eta^{\mathrm{NLoS}},
\end{align}
\end{subequations}
where $L^{\mathrm{FS}}_{m,l}(t)$ denotes the free space pathloss with  $L^{\mathrm{FS}}_{m,l}(t) =20\log(d_{m,l}(t)) + 20 \log(f)+20\log(\frac{4\pi}{c})$, and $f$ is the carrier  frequency.  Furthermore, $\eta^{\mathrm{LoS}}$ and $\eta^{\mathrm{NLoS}}$ are the mean additional  losses for LoS and
NLoS, respectively. Therefore, at time slot $t$, the  average pathloss between UAV $m$ and user $l$ can be expressed as
\begin{align}
L_{m,l}(t) = P^{\mathrm{LoS}}(t) \cdot PL_{m,l}^{\mathrm{LoS}}(t) + P^{\mathrm{NLoS}}(t) \cdot PL_{m,l}^{\mathrm{NLoS}}(t).
\end{align}

\subsubsection{LoS Model}
As discussed in \cite{Zeng16TC}, the LoS model provides  a good approximation for practical UAV-to-ground communications. In the LoS model, the path loss between a UAV and a ground user relies on the locations of the UAV and the ground user as well as the  type of propagation.  
Specifically, under the LoS model, the channel gains between the UAVs and the users follow the free space path loss model, which is determined by the distance between the UAV and the user. Therefore, at time slot $t$, the LoS channel power gain from the $m$-th UAV to the $l$-th ground user can be expressed as
\begin{align}\label{LosCh.eq}
g_{m,l}(t)=\beta_0d_{m,l}^{-\alpha}(t) = \frac{\beta_0}{\big(\|\mathbf{v}_l - \mathbf{u}_m(t)\|^2 + H_m^2\big)^{\frac{\alpha}{2}}},
\end{align}
where $\mathbf{u}_m(t)=(x_m(t),y_m(t))$, and $(x_m(t),y_m(t))$ denotes the location of UAV $m$ in the horizontal dimension at time slot $t$. Correspondingly, $\mathbf{v}_l = (x_l,y_l)$ denotes the location of user $l$.  Furthermore,  $\beta_0$ denotes the channel power gain at the reference distance of $d_0=1$ m, and $\alpha \geq 2$ is the path loss exponent.

\subsection{Signal Model}
In the UAV-to-ground transmission, the interference to each UAV-to-ground user pair is created by other UAVs operating on the same subchannel.   Let $c_m^{k}(t)$ denote the indicator of subchannel, where  $c_m^{k}(t)=1$ if  subchannel $k$ occupied by UAV $m$ at time slot $t$; $c_m^{k}(t)=0$, otherwise.  It satisfies
\begin{align}
\sum_{k\in\mathcal{K}} c_m^{k}(t) \leq 1.
\end{align}
That is each UAV can only occupy a single subchannel for each time slot. 
Let $a_m^l(t)$ be the indicator  of users.  $a_m^l(t) =1$ if user $l$ served by UAV $m$ in time slot $t$; $a_m^l(t)=0$, otherwise.    Therefore, the  observed signal-to-interference-plus-noise ratio (SINR) for a UAV-to-ground user communication between UAV $m$ and user $l$ over  subchannel $k$ at time slot $t$ is given by 
\begin{align}\label{SINR.eq1}
\gamma_{m,l}^{k}(t) =  \frac{G_{m,l}^{k}(t)a_m^l(t)c_m^{k}(t)P_m(t)}{I_{m,l}^{k}(t) + \sigma^2},
\end{align}
where $G_{m,l}^{k}(t)$ denotes the channel gain between  UAV $m$ and user $l$ over  subchannel $k$ at time slot $t$. $P_m(t)$ denotes the transmit power selected by UAV $m$ at time slot $t$. $I_{m,l}^{k}(t)$ is the interference to UAV $m$ with $I_{m,l}^{k}(t) = \sum_{j \in \mathcal{M}, j\neq m} G_{j,l}^{k}(t)c_{m}^{k}(t)P_j(t) $.  Therefore, at any time slot $t$, the SINR for UAV $m$ can be expressed as 
\begin{align}
\gamma_m(t) = \sum_{l \in \mathcal{L}} \sum_{k \in \mathcal{K}} \gamma_{m,l}^{k}(t).
\end{align}

In this article,   discrete transmit power control is adopted at  UAVs \cite{Zheng14TWC}.  The transmit power values by each UAV to communicate with its respective connected user can be expressed as a vector $\mathbf{P}=\{P_1,\cdots,P_J\}$. For each  UAV $m$, we define a binary variable $p_m^j(t)$, $j\in \mathcal{J}=\{1,\cdots,J\}$. $p_m^j(t)=1$, if UAV $m$ selects to transmit at a power level $P_j$ at time slot $t$; and $p_m^j(t)=0$, otherwise. Note that only one power level can be selected at each time slot $t$ by UAV $m$, we have
\begin{align}
\sum_{j \in \mathcal{J}} p_m^j(t) \leq 1, \forall m \in \mathcal{M}.
\end{align}

As a result, we can define a finite set of possible power level selection decisions made by UAV $m$, as follows.
 \begin{align}
 \mathcal{P}_m=\{p_m(t)\in \mathbf{P}|\sum_{j \in \mathcal{J}} p_{m}^j(t) \leq 1 \}, ~\forall m \in \mathcal{M}.
 \end{align}
 Similarly, we also define finite sets of all possible subchannel selection and user selection by UAV $m$, respectively, which are given as follows:
 \begin{align}
 \mathcal{C}_m=&\{c_m(t)\in \mathcal{K}| \sum_{k \in \mathcal{K}} c_{m}^k(t) \leq 1\}, \forall m \in \mathcal{M},\\
  \mathcal{A}_m=&\{a_m(t)\in \mathcal{L}| \sum_{l \in \mathcal{L}} a_{m}^l(t) \leq 1\}, \forall m \in \mathcal{M}.
 \end{align}

To proceed further, we assume that the considered multi-UAV network operates on a discrete-time basis where the time axis is partitioned into equal non-overlapping time intervals  (slots). Furthermore, the communication parameters are assumed to remain constant during each time slot.  
 Let $t$ denote an integer valued time slot index. Particularly,  each UAV holds the CSI of all ground users and decisions for a fixed time interval $T_s \geq 1$ slots, which is called decision period. We consider the following scheduling strategy for the transmissions of UAVs: Any UAV  is assigned a time slot $t$ to start its transmission and  must finish its transmission and select the new strategy or reselect the old strategy by the end of its decision period, i.e., at slot $t+T_s$.  We also assume that the UAVs do not know the accurate duration of their stay in the network. This feature motivates us to design an on-line learning algorithm  for optimizing the long-term  energy-efficiency performance of multi-UAV networks.

\section{Stochastic Game Framework for Multi-UAV Networks}
In this section, we first describe the optimization problem investigated in this article. Then, to model the uncertainty of  stochastic environments,  we formulate the problem of joint user, power level and subchannel selections  by UAVs to be a stochastic game.

\subsection{Problem Formulation}
Note that from \eqref{SINR.eq1} to achieve the maximal throughput, each UAV   transmits at a maximal power level, which, in turn, results in increasing interference to other UAVs.  Hence, to provide reliable communications of UAVs, the main goal of the dynamic design for joint user, power level  and  subchannel selection is to ensure that the  SINRs provided by the UAVs no less than the predefined thresholds. Specifically, the mathematical form can be expressed as
\begin{align}\label{QoS.eq}
\gamma_{m}(t) \geq \bar{\gamma}, \forall m \in \mathcal{M},
\end{align}    
where $\bar{\gamma}$ denotes the targeted QoS threshold of users served by UAVs.  
At  time slot $t$, if the constraint \eqref{QoS.eq} is satisfied, then the UAV obtains a reward $R_{m}(t)$, defined as the difference between the throughput and the cost of power consumption achieved by the selected user,  subchannel and power level. 
Otherwise, it receives a zero reward. Therefore, we can express the reward function $R_m(t)$ of  UAV $m$ at time slot $t$, as follows:
\begin{align}\label{RewardFunc.eq}
R_m(t) = \begin{cases}
\frac{W}{K}\log(1+\gamma_m(t)) - \omega_m P_m(t), & \text{if}~ \gamma_{m}(t) \geq \bar{\gamma}_m,\\
0, & \text{o.w.},
\end{cases}
\end{align}
for all $m \in \mathcal{M}$ and the corresponding immediate reward is denoted as $R_m(t)$. In \eqref{RewardFunc.eq}, $\omega_m$ is  the cost per unit level of power. Note that at any time slot $t$, the instantaneous reward of UAV $m$ in \eqref{RewardFunc.eq} relies on: 1) the observed information:  the individual user, subchannel and power level decisions of UAV $m$, i.e., $a_m(t)$, $c_m(t)$ and $p_m(t)$. In addition, it also relates with the current channel gain $G_{m,l}^k(t)$; 2) unobserved information: the subchannels and power levels selected by other UAVs and the channel gains.  
It should be pointed out that we omitted the fixed power consumption for UAVs, such as the power consumed by controller units and data processing \cite{Uragun11PowCons}.

Next, we  consider to maximize the long-term reward $v_m(t)$ by selecting the served user, subchannel and  transmit power level at each time slot. Particularly, we adopt a future discounted reward \cite{shoham2008multiagent}  as the measurement for each UAV. Specifically,   at a certain time slot of the process, the discounted reward is the sum of its payoff in the present time slot, plus the sum of future rewards discounted by a constant factor.  
Therefore, the considered long-term reward of UAV $m$ is given by
\begin{align}\label{LTReward.eq}
v_m(t)=\sum_{ \tau = 0}^{+\infty}\delta^{\tau} R_{m}(t+\tau+1),
\end{align}
where $\delta$ denotes the discount factor with $0\leq \delta < 1$.    Specifically,  values of $\delta$ reflect the effect of future rewards on the optimal decisions:  if  $\delta$ is close to 0, it means that the decision  emphasizes the near-term gain; By contrast, if  $\delta$ is close to 1, it gives more weights to future rewards and we say the decisions are farsighted.

Next we introduce the set of all possible user, subchannel and power level decisions made by UAV $m$, $m \in \mathcal{M}$, which can be denoted as $\Theta_m =\mathcal{A}_m \otimes \mathcal{C}_m \otimes \mathcal{P}_m$ with $\otimes$ denoting the Cartesian product. Consequently, the objective of each UAV $m$ is to make a selection $\theta_m^*(t) = (a_m^*(t),c_m^*(t), p_m^*(t)) \in \Theta_m$, which maximizes its long-term reward in \eqref{LTReward.eq}.  Hence the optimization problem for  UAV $m$, $m \in \mathcal{M}$, can be formulated as 
\begin{align}\label{Opt1.eq}
\theta_m^*(t) = {\mathrm{arg} \max}_{\theta_m \in \Theta_m} R_m(t).
\end{align}
Note that the optimization design for the considered multi-UAV network consists of $M$ subproblems, which corresponds to $M$ different UAVs. Moreover, each UAV has no information about other UAVs such as their rewards, hence one cannot solve problem \eqref{Opt1.eq} accurately.  
To  solve the optimization problem  \eqref{Opt1.eq} in  stochastic environments, we try to formulate the problem of joint user, subchannel and power level selections by UAVs  to a stochastic non-cooperative game in the following subsection.

\subsection{Stochastic Game Formulation}

In this subsection, we consider to model the formulated  problem \eqref{Opt1.eq} by adopting a  stochastic  game (also called Markov game) framework \cite{nowe2012game}, since it is the generalization of the Markov decision processes to the multi-agent case.  

In the considered network, $M$ UAVs communicate to  users with having no information about  the  operating environment.  It is assumed that all UAVs are selfish and rational. Hence, at any time slot $t$, all UAVs  select their actions non-cooperatively to maximize the long-term rewards in \eqref{Opt1.eq}.  Note that the action for each UAV $m$ is selected from its action space $\Theta_m$. The action conducted by UAV $m$ at time slot $t$, is a triple $\theta_m(t) = (a_m(t),c_m(t), p_m(t)) \in \Theta_m$, where  $a_m(t),~c_m(t)$ and  $p_m(t)$ represent the selected user, subchannel and power level respectively, for UAV $m$ at time slot $t$. For each UAV $m$, denote by $\theta_{-m}(t)$ the actions conducted by the other $M-1$ UAVs at time slot $t$, i.e.,  $\theta_{-m}(t) \in \Theta \setminus \Theta_m$.

As a result, the instantaneous SINR of UAV $m$ at time slot $t$ can be rewritten as 
\begin{eqnarray}\label{SINRF.eq}
\begin{aligned}
 \gamma_m(t)[\theta_m(t), \theta_{-m}(t),\mathbf{G}_m(t)] =  \sum_{l \in \mathcal{L}}\sum_{k \in \mathcal{K}} \frac{S_{m,l}^k(t)[\theta_{m}(t),\theta_{-m}(t), \mathbf{G}_{m,l}(t)] }{I_{m,l}^{k}(t)[\theta_{m}(t),\theta_{-m}(t),\mathbf{G}_{m,l}(t)] + \sigma^2},
\end{aligned}
\end{eqnarray}
where $S_{m,l}^k(t)=G_{m,l}^{k}(t)a_m^l(t)c_m^{k}(t)P_m(t)$, and $I_{m,l}^{k}(t)(\cdot)$ is given in \eqref{SINR.eq1}. Furthermore, $\mathbf{G}_{m,l}(t)$ denotes the matrix of  instantaneous channel responses between UAV $m$ and user $l$ at time slot $t$, which can be expressed as 
\begin{align}
\mathbf{G}_{m,l}(t)=\begin{bmatrix}
G_{1,l}^1(t) & \cdots &  G_{1,l}^K(t)\\
\vdots     & \ddots & \vdots  \\
G_{M,l}^1(t) & \cdots & G_{M,l}^K(t)
\end{bmatrix},
\end{align}  
with $\mathbf{G}_{m,l}(t) \in \mathbb{R}^{M \times K}$, for all $l \in \mathcal{L}$ and $m \in \mathcal{M}$.

At any time slot $t$, each UAV $m$ can measure its current SINR level $\gamma_{m}(t)$. Hence,   the sate  $s_m(t)$ for each UAV $m$, $m \in \mathcal{M}$, is fully observed, which can be defined as 
\begin{align}
s_m(t) = \begin{cases}
1 , & \text{if} ~\gamma_m(t) \geq \bar{\gamma}, \\
0, & \text{o.w.}.
\end{cases}
\end{align}
Let $\mathbf{s}=(s_1,\cdots,s_M)$ be a state vector for all UAVs. In this article, UAV $m$ does not know the states for other UAVs as UAV cannot cooperate with each other. 

We assume that the actions for each UAV satisfy the properties of  Markov chain, that is the reward of a UAV is only dependant on the current  state and action. As discussed in \cite{neyman2003markov}, Markov chain is used to describes the dynamics of the states of a stochastic game where each player has a single action in each state. Specifically, the formal definition of Markov chains is given as follows.
\begin{definition}
 A finite state Markov chain is  a discrete stochastic process, which can be described as follows: Let a finite set of states $ \mathcal{S} = \{s_1, \cdots, s_q\}$ and a $q \times q$ transition matrix $\mathbf{F}$ with each entry $0\leq F_{i,j} \leq 1$ and $\sum_{j=1}^q F_{i,j} = 1$ for any $1 \leq i \leq q$. The process starts in one of the states and moves to another state successively. Assume that the chain is currently in state $s_i$. The probability of  moving to the next state $s_j$  is
\begin{align}
 \Pr\{s(t+1)=s_j|s(t)=s_i\} = F_{i,j},
\end{align}   
which depends only on the present state and not on the previous states and is also called Markov property.  
\end{definition}

Therefore, the reward function of UAV $m$,   $m \in \mathcal{M}$, can be expressed as
\begin{eqnarray} \label{RewardFunc.eq2}
\begin{aligned}
r_m^t &= R_m(\theta_{m}^t,\theta_{-m}^t,s_m^t) \\
      &=s_m^t \Big(C_m^t[\theta_m^t, \theta_{-m}^t,\mathbf{G}_m^t] - \omega_m P_m[\theta_m^t] \Big).
\end{aligned} 
\end{eqnarray}
Here we put the time slot index $t$ in the superscript for notation compactness and  it is adopted in the following of this article for notational simplicity. In \eqref{RewardFunc.eq2}, the instantaneous transmit power is a function of the action $\theta_m^t$ and  the instantaneous rate of UAV $m$ is given by
\begin{eqnarray}
\begin{aligned}
C_m^t(\theta_m^t, \theta_{-m}^t,\mathbf{G}_m^t)
 =\frac{W}{K}\log\Big(1 + \gamma_{m}(\theta_m^t, \theta_{-m}^t,\mathbf{G}_m^t ) \Big),
\end{aligned}
\end{eqnarray}

Notice that from \eqref{RewardFunc.eq2}, at any time slot $t$, the reward $r_m^t$ received by UAV $m$ depends on the current state $s_m^t$, which is fully observed, and partially-observed actions $(\theta_m^t,\theta_{-m}^t)$. At the next time slot $t+1$,  UAV $m$ moves to a new random state $s_{m}^{t+1}$ whose possibilities are only based on the previous state $s_m(t)$ and the selected actions  $(\theta_m^t,\theta_{-m}^t)$.  This procedure repeats for the indefinite number of slots. Specifically, at any time slot $t$,  UAV $m$ can observe its state $s_m^t$ and the corresponding action $\theta_{m}^t$, but it does not know the actions of other players, $\theta_{-m}^t$, and the precise values $\mathbf{G}_{m}^t$.  The state transition probabilities are also unknown to each player UAV $m$. Therefore, the considered UAV system can be formulated as a stochastic game \cite{neyman2003stochastic}. 
\begin{definition}
A stochastic game can be defined as a tuple $\Phi = (\mathcal{S},\mathcal{M},\Theta,F, \mathcal{R})$ where:
\begin{itemize}
\item $\mathcal{S}$ denotes the  state set with $\mathcal{S}=\mathcal{S}_1\times \cdots\times \mathcal{S}_M$, $\mathcal{S}_m \in \{0,1\}$, for all $m\in \mathcal{M}$; 
\item  $M$ is the set of players;
\item $\Theta$ denotes the joint action set and $\Theta_m$ is the action set of player UAV $m$;
\item $F$ is the state transition probability function which depends on the actions of all players. Specifically, $F(s_m^t,\theta,s_m^{t+1})= \Pr\{s_m^{t+1} | s_m^t, \theta\}$, denotes the probability of transitioning to the next state $s_m^{t+1}$ from the state $s_m^t$ by  executing the joint action $
\theta$ with $\theta = \{\theta_1, \cdots,\theta_M\} \in \Theta$;
\item $\mathcal{R}=\{R_1,\cdots,R_M\}$, where $R_m: \Theta \times \mathcal{S} \rightarrow \mathbb{R}$ is a real valued  reward function for player $m$. 
\end{itemize}

\end{definition} 
 
In a stochastic game,  a  mixed strategy  $\pi_m$: $\mathcal{S}_m \rightarrow \Theta_m$,  denoting the mapping from the state set to the action set, is a collection of probability distribution over the available actions. Specifically, for UAV $m$ in the state $s_m$, its mixed strategy  is $\pi_m(s_m) = \{\pi_m(s_m,\theta_m)|\theta_m \in \Theta_m\}$, where each element $\pi_m(s_m,\theta_m)$ of $\pi_m(s_m)$ is the probability with UAV $m$ selecting an action $\theta_m$ in state $s_m$.  A joint strategy  $\pi=\{\pi_1(s_1),\cdots,\pi_{M}(s_M)\}$ is a vector of strategies for $M$ players with one strategy for each player. Let $\pi_{-m} = \{\pi_1,\cdots,\pi_{m-1},\pi_{m+1},\cdots,\pi_{M}(s_M)\}$ denote the same strategy profile but without the strategy $\pi_m$ of player  UAV $m$.  Based on the above discussions, the optimization goal of each player UAV $m$ in the formulated stochastic game is to maximize its expected reward over time. Therefore, for player UAV $m$ under a joint strategy $\pi = (\pi_1,\cdots,\pi_m)$ with assigning a strategy $\pi_i$ to each  UAV $i$, the optimization objective in \eqref{LTReward.eq} can be reformulated as 
\begin{eqnarray}
\begin{aligned}\label{LTReward.eq2}
V_m(s,\pi)= E \bigg\{ \sum_{ \tau = 0}^{+\infty}\delta^{\tau} r_{m}^{t+\tau+1} \mid  s^t=s \bigg\},
\end{aligned}
\end{eqnarray}
where $r_{m}^{t+\tau+1}$ represents the immediate reward received by UAV $m$  at time $t+\tau+1$ and $E\{ \cdot \}$ denotes the expectation operations. In the formulated stochastic game, players (UAVs)  have individual expected  reward  which depends on the joint strategy and not on the individual strategies of the players.  
Hence one cannot simply expect players to  maximize their expected rewards  as it may not be possible for all players to achieve this goal at the same time. Next, we describe a solution for the stochastic game by Nash equilibrium \cite{osborne1994course}. 
\begin{definition}
A Nash equilibrium is a collection of strategies, one for each player, so that each individual strategy  is a best-response to the others.   That is  if a solution $\pi^*=\{\pi_1^*,\cdots,\pi_M^*\}$ is a Nash equilibrium, then  for each UAV $m$, the strategy $\pi_m^*$ such that
\begin{eqnarray}
\begin{aligned}
V_m(\pi_m^*,\pi_{-m}) \geq V_m(\pi_m',\pi_{-m}), ~\forall \pi_m'.
\end{aligned}
\end{eqnarray}
\end{definition} 
It means that in a Nash equilibrium, each UAV's action is the best response to other UAVs' choice. Thus, in a Nash equilibrium solution, no UAV can benefit by changing its strategy as long as all the other UAVs keep their strategies constant. 
Note that the presence of imperfect information in the formulated non-cooperative stochastic  game provides opportunities for the players to learn their optimal strategies through repeated interactions with the stochastic environment. Hence, each player UAV $m$ is regarded as a learning agent whose task is to find a Nash equilibrium strategy for any state $s_m$.  In next section, we propose a multi-agent reinforcement-learning framework for maximizing the sum expected reward in \eqref{LTReward.eq2} with partial observations.

\section{Proposed Multi-Agent Reinforcement-Learning Algorithm}
In this section, we first describe the proposed MARL framework for multi-UAV networks. Then a Q-Learning based resource allocation algorithm  will be proposed for maximizing the expected long-term reward of the considered for multi-UAV network.

\subsection{MARL Framework for  Multi-UAV Networks}
\begin{figure} [t!]
\centering
\includegraphics[width= 3.5in, height=2.0in]{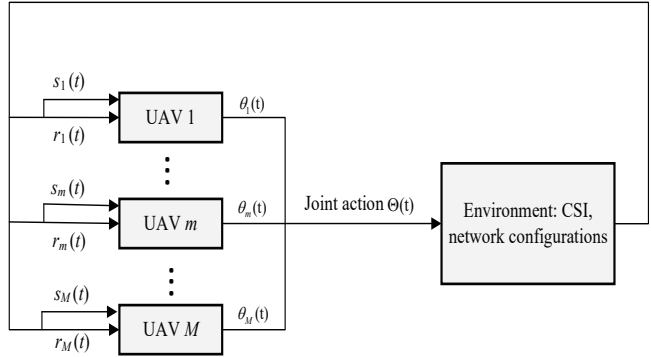}
  \vspace*{-0.5em}\caption{Illustration of MARL framework for multi-UAV networks.}
   \label{Framework-MARL-UAV}
   \vspace{-1.0em}
\end{figure}

Fig. \ref{Framework-MARL-UAV} describes the key components of MARL studied in this article. Specifically, for each UAV $m$, the left-hand side of the box is the locally observed information at time slot $t$--state $s_m^t$ and reward $r_m^t$; the right-hand side of the box is the action for UAV $m$ at time slot $t$. 
 The decision problem faced by a player in a stochastic game when all other players choose a fixed  strategy profile is equivalent to an Markov decision processes (MDP) \cite{neyman2003markov}.  An agent-independent method  is proposed,  for which all agents conduct a decision algorithm independently but share a common structure based on Q-learning. 
 
Since Markov property is used to model the dynamics of the environment, the rewards of UAVs are based only on the current state and action. MDP for agent (UAV) $m$ consists 
 of: 1) a discrete set of environment state $\mathcal{S}_m$, 2) a discrete set of possible actions $\Theta_m$, 3) a one-slot dynamics of the environment given by the state transition probabilities $F_{s_m^t \rightarrow s_m^{t+1}} = F(s_m^{t},\theta, s_m^{t+1})$ for all $\theta_m \in \Theta_m$ and $s_m^{t}, s_m^{t+1} \in \mathcal{S}_m$; 4) a reward function $R_m$ denoting the expected value of the next reward for UAV $m$. For instance, given the current state $s_m$, action $\theta_m$ and the next state $s_m'$: $R_m(s_m,\theta_m,s_m') = E\{r_m^{t+1}| s_m^t=s_m, \theta_m^t = \theta_m,s_m^{t+1} = s_m'\}$, where $r_m^{t+1}$ denotes the immediate reward of the environment to UAV $m$ at time $t+1$.  Notice that  UAVs cannot interact with each other, hence each UAV knows imperfect information of its operating stochastic environment.   In this article, Q-learning is  used to  solve MDPs, for which a learning agent operates in an unknown stochastic environment and does not know the   reward and transition functions \cite{suttonRL}. Next  we describe the  Q-learning algorithm for solving the MDP for one UAV. Without loss of generalities, UAV $m$ is considered for simplicity.  
Two fundamental concepts of algorithms for solving the above MDP is the state value function   and action value function (Q-function) \cite{neto2005single}. Specifically,  the former  in fact is the expected reward for some state in \eqref{LTReward.eq2} giving the agent in following some policy. Similarly, the Q-function for UAV $m$ is the expected reward starting from the state $s_m$,  taking  the action $\theta_m$ and following policy $\pi$, which can be expressed as
\begin{eqnarray}\label{Q-func.eq}
\begin{aligned}
Q_m(s_m,\theta_m,\pi)= 
E\bigg\{  \sum_{ \tau = 0}^{+\infty}\delta^{\tau} r_{m}^{t+\tau+1} \mid  s^{t}=s, \theta_m^{t}=\theta_m \bigg\},
\end{aligned}
\end{eqnarray}\noindent 
where the corresponding  values  of \eqref{Q-func.eq}  are called action values (Q-values).

\begin{proposition}\label{Pro:StateValuFunc}
 A recursive relationship for the state value function can be derived from the established return. Specifically, for any strategy $\pi$ and any state $s_m$, the following  condition holds between two consistency states  $s_m^t=s_m$ and $s_m^{t+1}=s'_m$, with  $s_m,~s'_m\in\mathcal{S}_m$: 
\begin{eqnarray} \label{LTRewardEx.eq2}
\begin{aligned}
 V_m(s_m,\pi) &= E\left\lbrace \sum_{\tau = 0}^{+\infty} \delta^{\tau} r_m^{t+\tau+1} | s_{m}^t  = s_m \right\rbrace \\
&= \sum_{s'_m \in \mathcal{S}_m} F(s_m,\theta,s'_m)\sum_{\theta \in \Theta} \prod_{j\in \mathcal{M}} \pi_j(s_j,\theta_j)   \times \left[  R_m(s_m,\theta,s'_m) + \delta V(s'_m,\pi) \right],
\end{aligned}
\end{eqnarray}
where $\pi_j(s_j,\theta_j)$ is the probability of choosing action $\theta_j$ in state $s_j$ for UAV $m$.
\begin{proof}
See Appendix~A. 
\end{proof}
\end{proposition}
Note that the state value function $V_m(s_m,\pi)$ is  the expected return when starting in state $s_m$ and following a strategy  $\pi$ thereafter.  
Based on {\bf Proposition \ref{Pro:StateValuFunc}}, we can rewrite the Q-function in \eqref{Q-func.eq} also into a recursive from, which is given by 
\begin{eqnarray} \label{Q-func.eq2}
\begin{aligned}
& Q_m(s_m,\theta_m,\pi)= E \bigg\lbrace r_m^{t+1} +  \delta \sum_{\tau = 0}^{+\infty} \delta^{\tau} r_m^{t+\tau+2}|s_m^t=s_m,\theta_m^t=\theta,s_m^{t+1}=s'_m \bigg\rbrace \\
&\quad = \sum_{s'_m\in \mathcal{S}_m} F(s_m,\theta,s'_m)  \sum_{\theta_{-m} \in \Theta_{-m}} \prod_{j\in \mathcal{M}\setminus \{m\}} \pi_j(s_j,\theta_j)   \times \left[R(s_m,\theta,s'_m) + \delta V_m(s'_m,\pi)\right].
\end{aligned}
\end{eqnarray}
Note that from \eqref{Q-func.eq2}, Q-values  depend on the actions of all the UAVs. 
It should be pointed out that \eqref{LTRewardEx.eq2} and \eqref{Q-func.eq2} are the basic equations for the Q-learning based reinforcement learning algorithm for solving the MDP of each UAV. 
From \eqref{LTRewardEx.eq2} and \eqref{Q-func.eq2}, we also can derive the following relationship between state values and Q-values:
\begin{eqnarray}\label{State-Q.eq}
\begin{aligned}
V_m(s_m,\pi)=\sum_{\theta_m \in \Theta_m}  \pi_m(s_m,\theta_m) Q_m(s_m,\theta_m,\pi).
\end{aligned}
\end{eqnarray}

As discussed above, the goal of solving a MDP is to find  an optimal strategy to obtain a maximal reward. An optimal strategy for UAV $m$ at state $s_m$,  can be defined, from the perspective of state value function, as
\begin{eqnarray}\label{OptStateVal.eq}
\begin{aligned}
V_m^* = \max_{\pi_m} V_m(s_m,\pi), ~~s_m \in \mathcal{S}_m.
\end{aligned}
\end{eqnarray} 
For the optimal Q-values, we also have
\begin{eqnarray}\label{OptQVal.eq}
\begin{aligned}
Q_m^*(s_m,\theta_m) = \max_{\pi_m} Q_m(s_m,\theta_m,\pi), s_m \in \mathcal{S}_m, \theta_m \in \Theta_m.
\end{aligned}
\end{eqnarray}
Substituting \eqref{State-Q.eq} to \eqref{OptStateVal.eq}, the optimal  state value equation in \eqref{OptStateVal.eq} can be reformulated as 
\begin{eqnarray}\label{OptStateVal.eq2}
\begin{aligned}
V_m^*(s_m) = \max_{\theta_m}  Q_m^*(s_m,\theta_m),
\end{aligned}
\end{eqnarray}
where the fact that $\sum_{\theta_m} \pi(s_m,\theta_m) Q_m^*(s_m,\theta_m) \leq \max_{\theta_m} Q_m^*(s_m,\theta_m)$ was applied to obtain \eqref{OptStateVal.eq2}. 
 Note that in \eqref{OptStateVal.eq2}, the optimal state value equation is a maximization over the action space instead of the strategy space. 

Next by combining \eqref{OptStateVal.eq2} with  \eqref{LTRewardEx.eq2} and \eqref{Q-func.eq2}, one can obtain the 
Bellman optimality equations, for state values and for Q-values, respectively:
\begin{eqnarray}\label{OptStateVal.eq3}
\begin{aligned}
&V_m^*(s_m)  = \sum_{\theta_{-m} \in \Theta_{-m}} \prod_{j\in \mathcal{M}\setminus \{m\}} \pi_j(s_j,\theta_j)  \times \max_{\theta_m} \sum_{s'_m} F(s_m,\theta,s'_m)    \left[ R(s_m,\theta_m,s'_m) + \delta V_m^*(s'_m) \right],
\end{aligned}
\end{eqnarray}\noindent
and
\begin{eqnarray}\label{Q-func.eq3}
\begin{aligned}
&Q_m^*(s_m,\theta_m) = \\& \sum_{\theta_{-m} \in \Theta_{-m}}  \prod_{j\in \mathcal{M}\setminus \{m\}} \pi_j(s_j,\theta_j) \times  \sum_{s'_m} F(s_m,\theta,s'_m)   \bigg[R(s_m,\theta_m,s'_m) + \delta \max_{\theta'_m} Q_m^*(s'_m,\theta'_m) \bigg].
\end{aligned}
\end{eqnarray}\noindent
Note that \eqref{Q-func.eq3} indicates that the optimal strategy will always choose an action that maximizes the Q-function for the current state.  In the multi-agent case, the Q-function of each agent depends on the joint action and is conditioned on the joint policy, which makes it complex to find an  optimal joint strategy \cite{neto2005single}.  To overcome these challenges, we consider   UAV are independent learners (ILs), that is UAVs do not observe the rewards and actions of the other UAVs, they interact with the environment as if no other UAUs exist.

\subsection{Q-Learning based Resource Allocation for  Multi-UAV Networks }

  In this subsection, an ILs \cite{matignon2012independent} based  MARL algorithm is proposed to solve the resource allocation among  UAVs. Specifically, 
   each UAV runs a standard Q-learning algorithm to learn its optimal Q-values and simultaneously determines an optimal strategy for the MDP. Specifically, the selection of an action in each iteration depends on Q-values in terms of two states- $s_m$ and its successors. Hence Q-values provide insights on the future quality of the actions in the successor state. The update rule for Q-learning \cite{suttonRL} is given by 
\begin{eqnarray} \label{Q_update.eq}
\begin{aligned}
&Q_m^{t+1}(s_m,\theta_m) = Q_m^{t}(s_m,\theta_m) +  \alpha^t   \bigg\lbrace r_m^t  + \delta \max_{\theta'_m \in \Theta_m} Q_m^{t}(s'_m,\theta'_{m}) - Q_m^t(s_m,\theta_m) \bigg\rbrace, 
\end{aligned}
\end{eqnarray}\noindent
with $s_m^t = s_m, ~\theta_{m}^t = \theta_m$, where $s'_m$ and $\theta'_m$ correspond to $s_m^{t+1}$ and $\theta_m^{t+1}$, respectively. 
Note that an optimal action-value function can be obtained recursively from the corresponding action-values. Specifically, each agent learns the optimal action-values based on the updating rule in \eqref{Q_update.eq}, where $\alpha_t$ denotes the learning rate and $Q_m^t$ is the action-value of UAV $m$ at time slot $t$.

Another important component of Q-learning is  action selection mechanisms, which are used to select the actions that the agent will perform during the learning process. Its purpose is to strike a balance between exploration and exploitation that the agent can reinforce the evaluation it already knows to be good but also explore new actions \cite{suttonRL}. In this article, we consider  $\epsilon$-greedy exploration.  
In $\epsilon$-greedy selection, the agent  selects a random action with probability $\epsilon$ and selects the best action, which corresponds to the highest Q-value at the moment, with probability $1 - \epsilon$.  As such, the probability of selecting action $\theta_m$ at state $s_m$ is given by
\begin{eqnarray}\label{eGreedy.eq}
\begin{aligned}
\pi_m(s_m,\theta_m) = \begin{cases}1 - \epsilon ,~\text{if}~ Q_m~\text{of}~\theta_m \text{is the highest}, \\
\epsilon, ~\text{otherwise}.
\end{cases}
\end{aligned}
\end{eqnarray}
where $\epsilon \in (0,1)$. 
To ensure the convergence of Q-learning, the learning rate $\alpha_t$  are set as in \cite{jaakkola1994convergence}, which is given by 
\begin{eqnarray}\label{LearnRate.eq}
\begin{aligned}
\alpha_t = \frac{1}{(t+c_{\alpha})^{\varphi_{\alpha}}},
\end{aligned}
\end{eqnarray}\noindent
where $c_{\alpha} >0$,  $\varphi_{\alpha} \in (\frac{1}{2},1]$ .

Note that each UAV runs the Q-learning procedure independently in the proposed ILs based MARL algorithm. Hence, for each UAV $m$, $m\in \mathcal{M}$, the  Q-learning procedure  is concluded in {\bf Algorithm \ref{alg:MARL-QL}}. In {\bf Algorithm \ref{alg:MARL-QL}}, the initial Q-values are  set to zero, therefore, it is also called zero-initialized Q-learning \cite{koenig1992complexity}. Since UAVs have no prior information on the initial state, a UAV takes a strategy with equal probabilities, i.e., $\pi_{m}(s_m,\theta_m) = \frac{1}{|\Theta_m|}$.

\begin{algorithm}
\caption{  Q-learning based MARL algorithm for UAVs}
\label{alg:MARL-QL}
\begin{algorithmic}[1]
\STATE {{\bf Initialization:}}
\STATE{Set $t=0$ and the parameters $\delta,~c_{\alpha}$}
\FORALL{ $m \in \mathcal{M}$}
\STATE{Initialize the action-value $Q_m^t(s_m,\theta_m)=0$, strategy $\pi_{m}(s_m,\theta_m) = \frac{1}{|\Theta_m|}=\frac{1}{MKJ}$;  }
\STATE{Initialize the state $s_m = s_m^t = 0$;}
\ENDFOR
\STATE{{\bf Main Loop:}}
\WHILE{$t < T$}
\FORALL{UAV $m$, $m\in\mathcal{M}$}
\STATE{Update  the learning rate $\alpha_t$  according to \eqref{LearnRate.eq}.}
\STATE{Select an action $\theta_m$ according to the strategy $\pi_m(s_m)$.} 
\STATE{Measure the achieved SINR at the receiver according to \eqref{SINRF.eq};}
\IF{$\gamma_m(t)\geq \bar{\gamma}_m$}
\STATE{Set $s_m^t=1$.}
\ELSE
\STATE{Set $s_m^t=0$. }
\ENDIF
\STATE{Update the instantaneous reward $r_{m}^t$ according to \eqref{RewardFunc.eq2}.}
\STATE{Update the action-value $Q_{m}^{t+1}(s_m,\theta_m)$ according to \eqref{Q_update.eq}. }
\STATE{Update the strategy $\pi_m(s_m,\theta_m)$ according to \eqref{eGreedy.eq}. }
\STATE{Update $t=t+1$ and the state $s_m = s_m^t$.}
\ENDFOR
\ENDWHILE
\end{algorithmic}
\end{algorithm}

\subsection{ Analysis of the proposed MARL algorithm}
In this subsection, we investigate the convergence  of  the proposed MARL based resource allocation algorithm. Notice that  the proposed MARL algorithm can be treated as an independent multi-agent Q-learning algorithm, in which each UAV as a learning agent makes a decision based on the Q-learning algorithm.
 Therefore, the convergence is concluded in the following proposition.  
 \begin{proposition}\label{Pro:Convergence}
 In the proposed MARL algorithm of {\bf Algorithm \ref{alg:MARL-QL}}, the Q-learning procedure for each UAV is always converged to the Q-value for individual optimal strategy. 
  \end{proposition}
  
The proof of {\bf Proposition \ref{Pro:Convergence}} depends on the following observations. Due to the non-cooperative property of UAVs, the convergence of the proposed MARL algorithm is dependent on the convergence of Q-learning algorithm   \cite{matignon2012independent}.  Therefore, we focus on the proof of convergence for the Q-learning algorithm in {\bf Algorithm \ref{alg:MARL-QL}}.  
\begin{theorem}\label{Theo:ConvergenceQ}
The Q-learning algorithm in {\bf Algorithm \ref{alg:MARL-QL}}  with the update rule in \eqref{Q_update.eq} converges with probability one (w.p.1)  to the optimal $Q_m^*(s_m,\theta_m)$ value if 
\begin{enumerate}
\item The state and action spaces are finite;
\item $\sum_{t=0}^{+\infty}\alpha^t = \infty$, $\sum_{t=0}^{+\infty}(\alpha^t)^2 < \infty$ uniformly w.p. 1;
\item $\mathrm{Var}\{r_m^t\}$ is bounded;
\end{enumerate} 
\end{theorem} 
 
\begin{proof}
See Appendix~B.

 \end{proof}

\section{Simulation Results}
In this section, we verify the effectiveness of the proposed MARL based resource allocation algorithm for multi-UAV networks by simulations. We consider  multi-UAV networks deployed in a disc area with a radius $r_d = 500$ m. The ground users are randomly and uniformly distributed inside the disk. All UAVs are assumed to fly at a fixed altitude $H = 100$ m. In the simulations, the noise power is assumed to be $\sigma^2 = -80$ dBm, the  subchannel bandwidth is $\frac{W}{K}=75$ KHz and $Ts = 0.1$ s. For the probabilistic model, the channel parameters in the simulations follow \cite{Alzenad17WCL}, where  $\mathfrak{a} = 9.61$ and $\mathfrak{b} = 0.16$. Moreover, the carrier frequency is $f = 2$ GHz, $\eta^{\mathrm{LoS}} = 1$ and $\eta^{\mathrm{NLoS}} = 20$. 
For the LoS channel model, the channel power gain at reference distance $d_0 = 1$ m is set as $\beta_0 = -60$ dB and the path loss coefficient is set as $\alpha=2$ \cite{Wu18TWC}.  In the simulations, the maximal power level number is $J=3$, the maximal power for each UAV is $P_m = P = 23$ dBm, where the maximal power is equally divided into $J$ discrete power values. The cost per unit level of power is $\omega_m = \omega = 100$ and 
the minimum SINR for the users is set as $\gamma_0 = 3$ dB. Moreover, $c_{\alpha} = 0.5$, $\rho_{\alpha} = 0.8$ and $\delta = 1$. 

In Fig. \ref{Ill_model.Fig}, we consider a random realization of  a multi-UAV network in horizontal plane, where $L = 100$ users are uniformly distributed in a disk with radius $r = 500$ m and two UAVs are initially  located at the edge of the disk with the angle $\phi = \frac{\pi}{4}$. For illustrative purposes, Fig. \ref{Com_1.Fig} shows the average  reward and the average reward per time slot of the UAVs under the setup of  Fig. \ref{Ill_model.Fig}, where the speed of the UAVs are set as $40$ m/s. Fig. \ref{Fig.sub.1:Com_avg_reward} shows  average  rewards with different $\epsilon$, which  is calculated as $v^t = \frac{1}{M}\sum_{m \in \mathcal{M}} v_m^t$. As can be  observed from Fig. \ref{Fig.sub.1:Com_avg_reward}, the average reward increases with the algorithm iterations. This is because the long-term reward can be improved by the proposed MARL algorithm. However, the curves of the average reward become flat when $t$ is higher that 250 time slots. In fact, the UAVs will fly outside  the disk when $t > 250$. As a result, the average reward will not increase. Correspondingly,  Fig. \ref{Fig.sub.2:avg_reward_per_epsiode} illustrates the average instantaneous  reward per time slot $r^t = \sum_{m \in \mathcal{M}} r_m^t$.  As can be observed from Fig. \ref{Fig.sub.2:avg_reward_per_epsiode},  the average reward per time slot   decreases with    algorithm iterations. This is because that the learning rate $\alpha_t$ in the adopted Q-learning procedure is a function of $t$ in \eqref{LearnRate.eq}, where $\alpha_t$ decreases with time slots increasing. 
 Notice that from  \eqref{LearnRate.eq}, $\alpha_t$ will decrease with algorithm iterations, which means that the update rate of the Q-values becomes slow with increasing $t$. Moreover, Fig. \ref{Com_1.Fig} also investigates the average reward with different $\epsilon = \{0,~0.2,~0.5,~0.9\}$.  If $\epsilon = 0$, each UAV will choose a greedy action which is also called exploit strategy.  If $\epsilon$ goes to 1, each UAV will choose a   random action with higher probabilities. Notice that from   Fig. \ref{Com_1.Fig}, $\epsilon = 0.5$ is a good choice in the considered setup.

\begin{figure} [!t]
\centering
\includegraphics[width= 3.5in, height=2.8in]{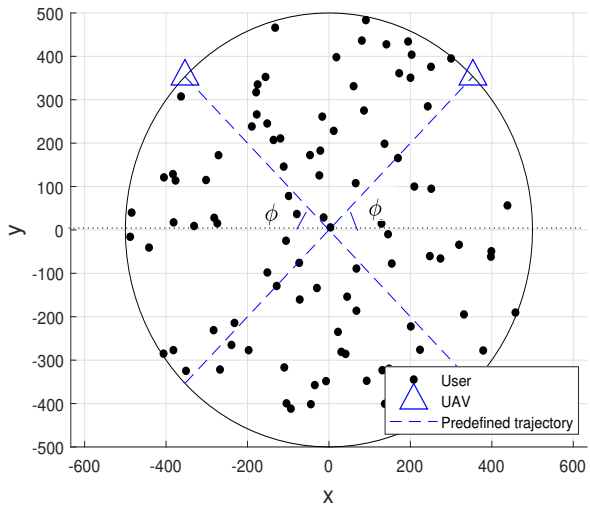}
 \caption{Illustration of UAVs based networks with $M = 2$ and $L=100$.}\label{Ill_model.Fig}
\end{figure}

\begin{figure}[t!]
\centering
\subfigure[Comparisons of average rewards.  ]{
\begin{minipage}[b]{0.45\textwidth}
\label{Fig.sub.1:Com_avg_reward}
\includegraphics[width=1\textwidth]{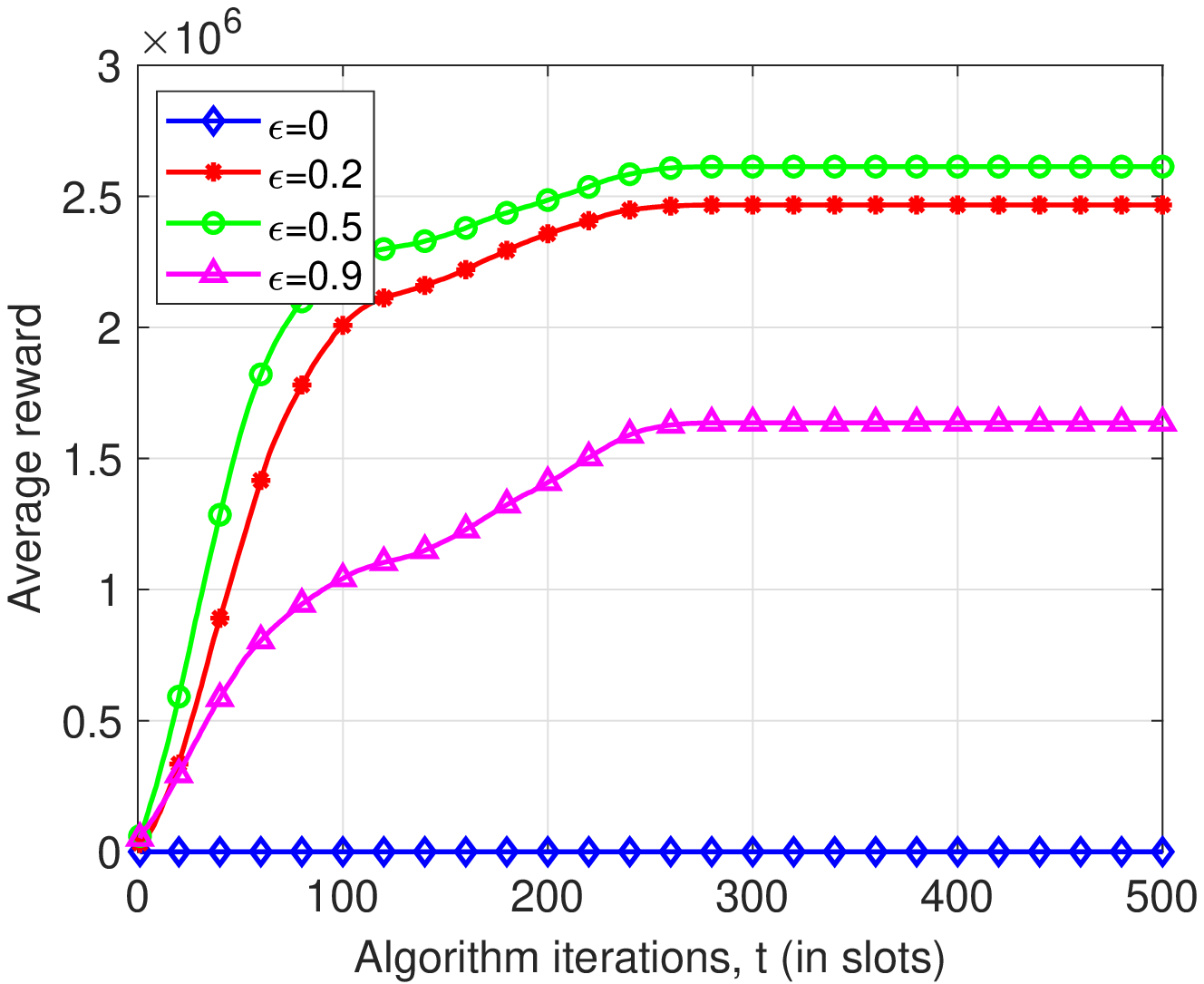}
\end{minipage}
}
\subfigure[Average rewards per time slot.]{
\label{Fig.sub.2:avg_reward_per_epsiode}
\begin{minipage}[b]{0.45\textwidth}
\includegraphics[width=1\textwidth]{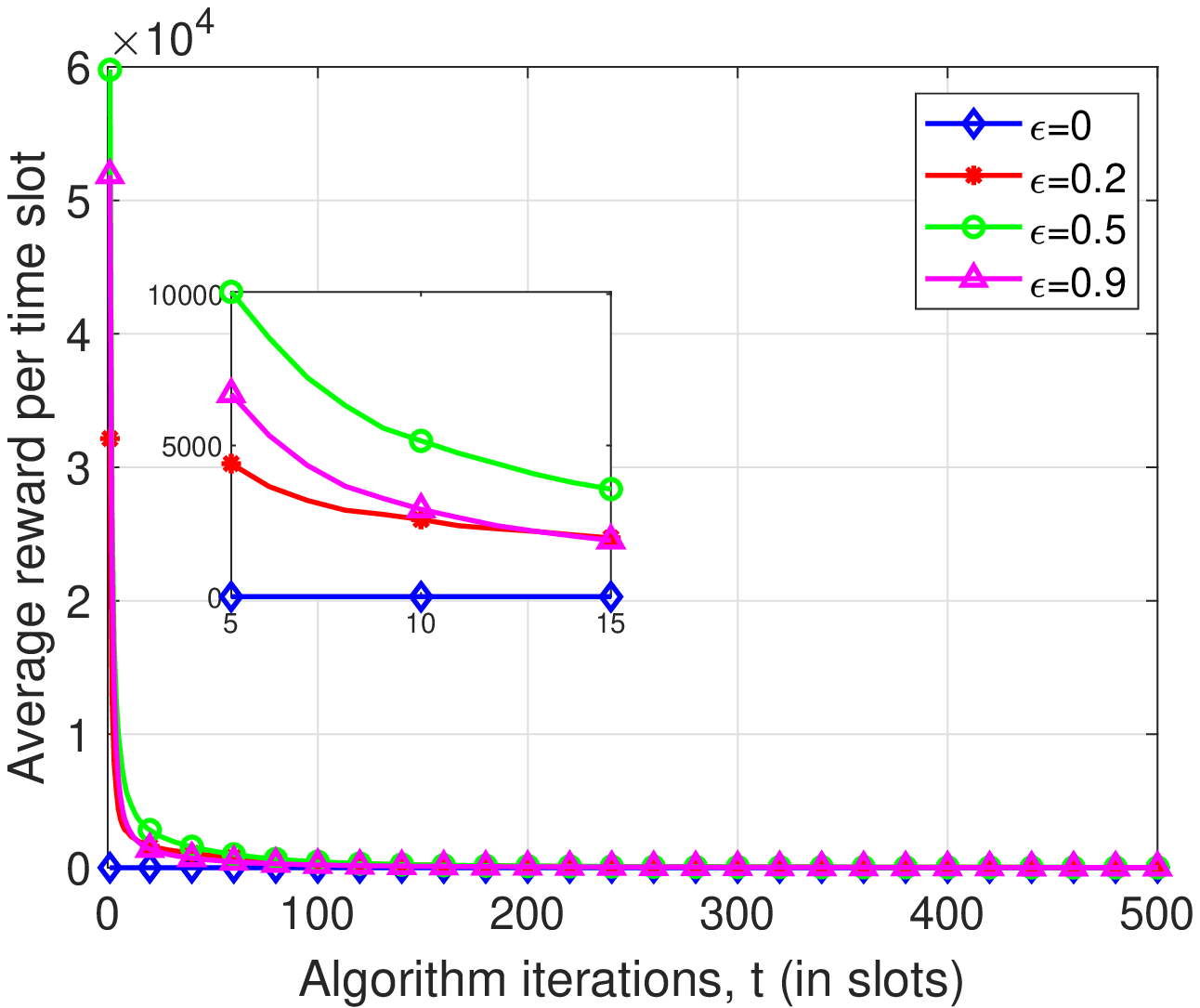}
\end{minipage}
}
\caption{Comparisons for average rewards with different $\epsilon$, where  $M = 2$ and $L=100$.}
\label{Com_1.Fig}
\end{figure}

In Fig. \ref{Com_2.Fig} and Fig. \ref{Com_3.Fig}, we investigate the average reward under different system configurations. Fig. \ref{Com_2.Fig} illustrates the average reward  with LoS channel model given in \eqref{LosCh.eq} over different $\epsilon$. Moreover, Fig. \ref{Com_3.Fig} illustrates the  average reward under probabilistic model with $M = 4,~K=3$ and $L=200$. Specifically, the UAVs randomly distributed in the cell edge. In the iteration procedure, each UAV flies over the cell followed by a straight line over the cell center, that is the center of the disk.  As can be observed from   Fig.  \ref{Com_2.Fig} and Fig. \ref{Com_3.Fig}, the curves of the average reward have the similar trends with that of Fig. \ref{Com_1.Fig} under different $\epsilon$. Besides, the considered multi-UAV network attains the optimal average reward when $\epsilon = 0.5$ under different network configurations.  
\begin{figure}[t!]
\centering
\subfigure[Comparisons of average rewards.  ]{
\begin{minipage}[b]{0.45\textwidth}
\label{Fig2.sub.1:Com_avg_reward}
\includegraphics[width=1\textwidth]{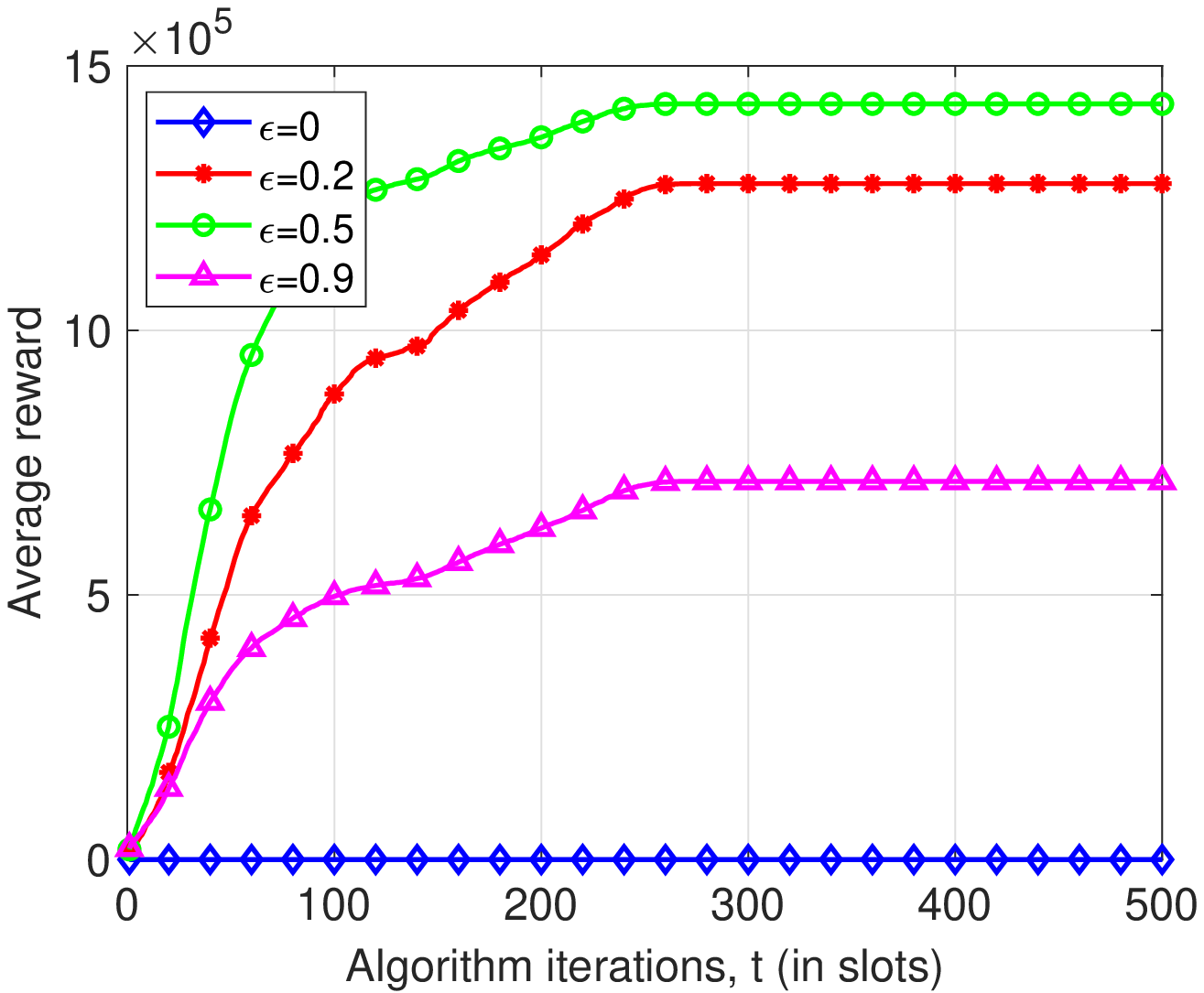}
\end{minipage}
}
\subfigure[Average rewards per time slot.]{
\label{Fig2.sub.2:avg_reward_per_epsiode}
\begin{minipage}[b]{0.45\textwidth}
\includegraphics[width=1\textwidth]{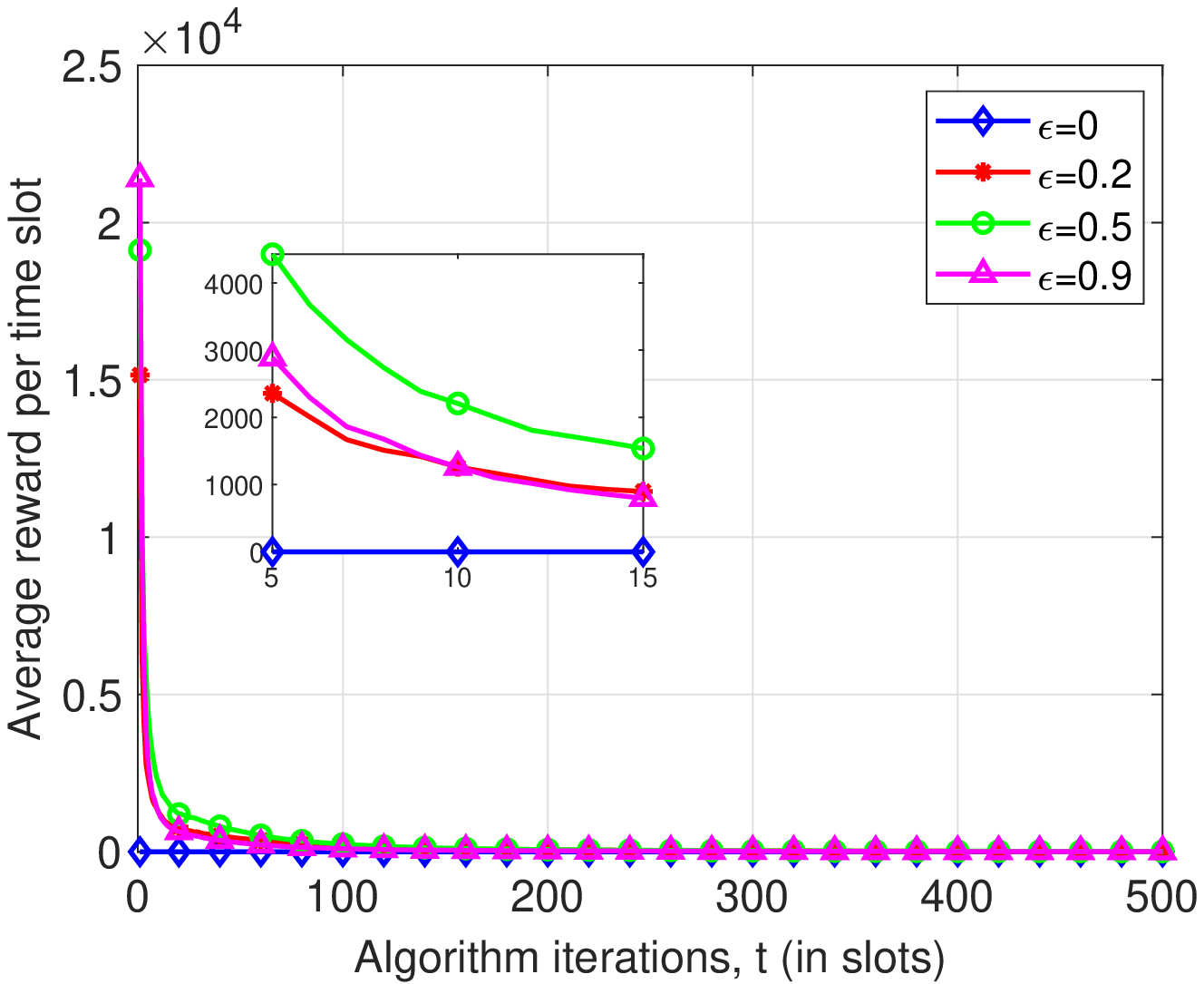}
\end{minipage}
}
\caption{ LoS channel model with different $\epsilon$, where $M = 2$ and $L=100$.}
\label{Com_2.Fig}
\end{figure}

\begin{figure} [!t]
\centering
\includegraphics[width= 3.5in, height=2.8in]{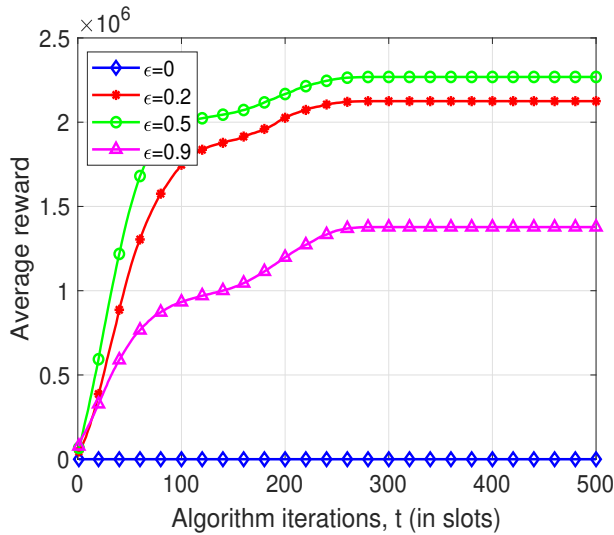}
 \caption{Illustration of multi-UAV networks with $M = 4,~K=3$ and $L=200$.}\label{Com_3.Fig}
\end{figure}

In Fig. \ref{Com4_Throughput.Fig}, we investigate the average reward of the proposed MARL algorithm by comparing it to the matching theory based resource allocation algorithm (Match).  In Fig. \ref{Com4_Throughput.Fig}, we consider the same setup as in Fig. \ref{Com_1.Fig}  but with $J=1$ for the simiplicity of algorithm implementation, which indicates that the UAV's action only contains the user selection for each time slot.  Furthermore, we consider  complete information exchanges among UAVs  are performed in the matching theory based user selection algorithm, that is each UAV knows other UAVs' action before making its own decision. comparisons, in the  matching theory based user selection procedure, we adopt the Gale-Shapley (GS) algorithm  \cite{gale1962college} at each time slot. Moreover, we also consider the performance of the randomly user selection algorithm (Rand) as a baseline scheme in Fig. \ref{Com4_Throughput.Fig}. As can be observed that from \ref{Com4_Throughput.Fig}, the achieved average reward of the matching based user selection algorithm  outperforms that of the  proposed MARL algorithm. This is because there is not information exchanges in the proposed MARL algorithm. In this case, each UAV cannot observe the other UAVs' information such as rewards and decisions, and thus it makes its decision independently. Moreover, it can be observed from  Fig. \ref{Com4_Throughput.Fig},  the average reward for the randomly user selection algorithm is lower than that of the proposed MARL algorithm. This is because of  the randomness of user selections, it cannot exploit the observed information effectively. As a result, the proposed MARL algorithm can achieve a tradeoff between reducing the information exchange overhead  and improving the system performance. 

\begin{figure} [!t]
\centering
\includegraphics[width= 3.5in, height=2.8in]{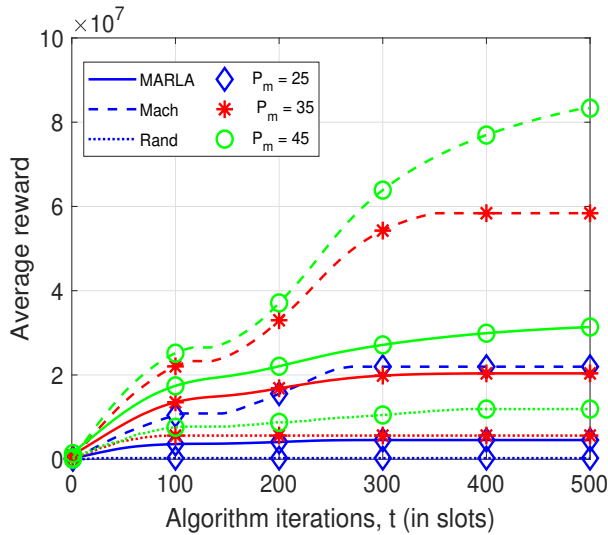}
 \caption{Comparisons of average rewards among different algorithms, where $M=2$, $K=1$, $J=1$ and $L=100$.}\label{Com4_Throughput.Fig}
\end{figure}

In Fig. \ref{3D.Fig}, we investigate the average reward as a function of  algorithm iterations  and the UAV's speed, where a UAU from a random initial location in the disc edge, flies over the disc along a direct line across the disc center with different speeds. The setup in Fig. \ref{3D.Fig} is the same as that in Fig. \ref{Com_3.Fig}  but with $M=1$ and $K=1$ for illustrative purposes. As can be observed that for a fixed speed, the average reward increases monotonically with increasing the algorithm iterations. Besides, for a fixed time slot, the average reward of larger speeds increases faster than that with  smaller speeds when $t$ is smaller than 150. This is due to the randomness of the locations for  users and the UAV, at the start point the UAV may not find an appropriate user satisfying its QoS requirement. Fig. \ref{3D.Fig} also shows that  the achieved average reward decreases when  the speed increases at the end of  algorithm iterations. This is because that if the UAV flies with a high speed, it will  take less time to fly out the disc.  As a result, the UAV with higher speeds  has less serving time than that of slower speeds. 

\begin{figure} [!t]
\centering
\includegraphics[width= 3.5in, height=2.8in]{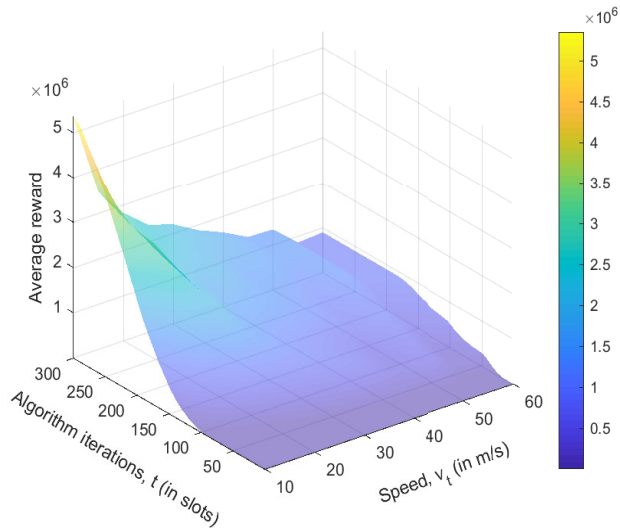}
 \caption{Average rewards with different time slots and speeds.}\label{3D.Fig}
\end{figure}

\section{Conclusions}

In this article, we  investigated the real-time designs of  resource allocation for multi-UAV downlink networks to maximize the long-term rewards. 
 Motivated by the uncertainty of environments, we proposed a stochastic game  formulation for the dynamic resource allocation problem of the considered multi-UAV networks, in which the goal of  each UAV was  to find a strategy of the resource allocation for maximizing its expected reward.  To overcome the  overhead of the information exchange and computation,   we developed an ILs based MARL algorithm to solve the  formulated stochastic game, where   all UAVs conducted a decision independently  based on Q-learning. Simulation results revealed that the proposed MARL based  resource allocation algorithm  for the multi-UAV networks  can attain a tradeoff between the information exchange overhead and the system performance. 
One promising extension of this work is to consider  more complicated joint  learning  algorithms for multi-UAV  networks with the partial information exchanges, that is the need of cooperation.
 Moreover,  incorporating the optimization of deployment and trajectories of UAVs into multi-UAV networks is capable of further improving energy efficiency of   multi-UAV networks, which is another promising future research direction.

\numberwithin{equation}{section}
\section*{Appendix~A: Proof of Proposition \ref{Pro:StateValuFunc}} \label{Appdx1}
\renewcommand{\theequation}{A.\arabic{equation}}
\setcounter{equation}{0}
Here, we show that the state values for one UAV $m$ over time in  \eqref{LTRewardEx.eq2}.  For one UAV $m$ with state $s_m \in \mathcal{S}_m$ at time step $t$, its state value function can be expressed as
\begin{eqnarray} \label{Proof.LTRewardEx.eq2}
\begin{aligned}
&V(s_m,\pi)=E\left\lbrace \sum_{\tau=0}^{+\infty} \delta^{\tau} r_m^{t+\tau+1}| s_m^t=s_m  \right\rbrace \\
&=E\left\lbrace r_m^{t+1} + \delta \sum_{\tau=0}^{+\infty} \delta^{\tau} r_m^{t+\tau+2}| s_m^t=s_m  \right\rbrace \\ 
&= E\left\lbrace r_m^{t+1} | s_m^t=s_m  \right\rbrace + \delta E\left\lbrace   \sum_{\tau=0}^{+\infty} \delta^{\tau} r_m^{t+\tau+2}| s_m^t=s_m  \right\rbrace,
\end{aligned}
\end{eqnarray}
where the first part and the second part represent the expected value and the state value function, respectively, at time $t+1$  over the state space and the action space. Next we show the relationship between the first part and the reward function $R(s_m,\theta,s_m')$ with $s_m^t = s_m,~\theta_m^t = \theta$ and $s_m^{t+1}=s'_m$. 
\begin{eqnarray}\label{SVPart1.eq}
\begin{aligned}
&E\left\lbrace r_m^{t+1} | s_m^t=s_m  \right\rbrace \\
&=\sum_{s'_m \in \mathcal{S}_m} F(s_m,\theta,s'_m)\sum_{\theta \in \Theta} \prod_{j\in \mathcal{M}} \pi_j(s_j,\theta_j) \times  E \left\lbrace  r^{t+1} | s_m^t = s_m, \theta_m^t = \theta_m, s_m^{t+1}=s'_m \right\rbrace \\
&=\sum_{s'_m \in \mathcal{S}_m} F(s_m,\theta,s'_m)\sum_{\theta \in \Theta} \prod_{j\in \mathcal{M}} \pi_j(s_j,\theta_j) R_m(s_m,\theta,s'_m),
\end{aligned}
\end{eqnarray}
where the definition of $R_m(s_m,\theta,s'_m)$ has been used  to obtain the final step. Similarly, the second part can be transformed into 
\begin{eqnarray}\label{SVPart2.eq}
\begin{aligned}
&E\left\lbrace   \sum_{\tau=0}^{+\infty} \delta^{\tau} r_m^{t+\tau+2}| s_m^t=s_m  \right\rbrace \\
&=\sum_{s'_m \in \mathcal{S}_m} F(s_m,\theta,s'_m)\sum_{\theta \in \Theta} \prod_{j\in \mathcal{M}} \pi_j(s_j,\theta_j) \times  E \left\lbrace  \sum_{\tau=0}^{+\infty} \delta^{\tau} r_m^{t+\tau+2}| s_m^t=s_m,\theta_m^t = \theta_m, s_m^{t+1}=s'_m  \right\rbrace\\
&=\sum_{s'_m \in \mathcal{S}_m} F(s_m,\theta,s'_m)\sum_{\theta \in \Theta} \prod_{j\in \mathcal{M}} \pi_j(s_j,\theta_j) V(s'_m,\pi).
\end{aligned}
\end{eqnarray}
Substituting \eqref{SVPart1.eq} and \eqref{SVPart2.eq} into \eqref{Proof.LTRewardEx.eq2}, we get
\begin{eqnarray}
\begin{aligned}
&V(s_m,\pi) = \sum_{s'_m \in \mathcal{S}_m} F(s_m,\theta,s'_m)\sum_{\theta \in \Theta} \prod_{j\in \mathcal{M}} \pi_j(s_j,\theta_j) \times  \left[  R_m(s_m,\theta,s'_m) + \delta V(s'_m,\pi) \right].
\end{aligned}
\end{eqnarray}
Thus, {\bf Proposition \ref{Pro:StateValuFunc}} is proved.

\numberwithin{equation}{section}
\section*{Appendix~B: Proof of Theorem \ref{Theo:ConvergenceQ}} \label{Appdx2}
\renewcommand{\theequation}{B.\arabic{equation}}
\setcounter{equation}{0}

The proof of {\bf Theorem \ref{Theo:ConvergenceQ}} follows from the idea in \cite{jaakkola1994convergence,melo2001convergence}. Here we give a more general procedure for  {\bf Algorithm  \ref{alg:MARL-QL}}. 
Note that the Q-learning algorithm is a stochastic form of value iteration \cite{jaakkola1994convergence}, which can be observed  from \eqref{Q-func.eq2} and \eqref{Q-func.eq3}. That is to perform a step of value iteration requires knowing the expected reward and the transition probabilities.  Therefore, to prove the convergence of the Q-learning algorithm,  stochastic approximation theory is applied. We first introduce a result of stochastic approcximation given in \cite{jaakkola1994convergence}.
\begin{lemma}\label{Lemma:StoApproxi}
A random iterative process $\bigtriangleup^{t+1}(x)$, which is defined as
\begin{align}
\bigtriangleup^{t+1}(x) =  (1-\alpha^t(x))\bigtriangleup^t(x) + \beta^t(x)\Psi^t(x),
\end{align}  
converges to zero w.p.1 if and only if the following conditions are satisfied.
\begin{enumerate}
\item The state space is finite;
\item $\sum_{t=0}^{+\infty}\alpha^t = \infty$, $\sum_{t=0}^{+\infty}(\alpha^t)^2 < \infty$, $\sum_{t=0}^{+\infty}\beta^t = \infty$, $\sum_{t=0}^{+\infty}(\beta^t)^2 < \infty$, and $E\{\beta^{t}(x)|\Lambda^{t} \} \leq E\{\alpha^{t}(x)|\Lambda^{t} \}$ uniformly w.p. 1;
\item $\|E\{\Psi^{t}(x)|\Lambda^{t} \}\|_{W} \leq \varrho \|\bigtriangleup^{t} \|_{W}$, where $\varrho \in (0,1)$;
\item $\mathrm{Var}\{\Psi^{t}(x)|\Lambda^{t} \} \leq C (1 + \|\bigtriangleup^{t} \|_{W})^2$, where $C>0$ is a constant.
\end{enumerate}
 Note that  $\Lambda^{t} = \{\bigtriangleup^{t},\bigtriangleup^{t-1},\cdots, \Psi^{t-1},\cdots,\alpha^{t-1},\cdots,\beta^{t-1}\}$ denotes the past at time slot $t$. $\| \cdot \|_W$ denotes some weighted maximum norm. 
\end{lemma}

Based on the results given in {\bf Lemma \ref{Lemma:StoApproxi}}, we now prove {\bf Theorem \ref{Theo:ConvergenceQ}} as follows.

Note that the Q-learning update equation in \eqref{Q_update.eq} can be rearranged as 
\begin{eqnarray} \label{Q_updateRe1.eq}
\begin{aligned}
Q_m^{t+1}(s_m,\theta_m) =& (1 - \alpha^t) Q_m^{t}(s_m,\theta_m) +   \alpha^t   \big\lbrace r_m^t  + \delta \max_{\theta'_m \in \Theta_m} Q_m^{t}(s'_m,\theta'_{m})  \big\rbrace. 
\end{aligned}
\end{eqnarray}\noindent
By subtracting $Q_m^*(s_m,\theta_m)$  from both side of \eqref{Q_updateRe1.eq}, we have 
\begin{eqnarray}\label{Q_updateRe2.eq}
\begin{aligned}
&\bigtriangleup_m^{t+1}(s_m,\theta_m) = (1 - \alpha^t) \bigtriangleup_m^{t}(s_m,\theta_m) + \alpha^t   \delta \Psi^{t}(s_m,\theta_m) ,
\end{aligned}
\end{eqnarray}
where 
\begin{align}
&\bigtriangleup_m^t(s_m,\theta_m) =Q_m^{t}(s_m,\theta_m)- Q_m^*(s_m,\theta_m), \label{Triangle.eq}\\
&\Psi_m^t(s_m,\theta_m)= r_m^t + \delta \max_{\theta'_m \in \Theta_m} Q_m^{t}(s'_m,\theta'_{m})  - Q_m^*(s_m,\theta_m) . \label{Psi.eq}
\end{align}
Therefore, the Q-learning algorithm can be seen as the random process of  {\bf Lemma \ref{Lemma:StoApproxi}} with $\beta^t = \alpha^t$.

Next we prove that the $\Psi^{t}(s_m,\theta_m)$ has the properties of 3) and 4) in {\bf Lemma \ref{Lemma:StoApproxi}}. 
We start by showing that $\Psi^{t}(s_m,\theta_m)$ is a contraction mapping  with respect to some maximum norm. 

\begin{definition}
For a set $\mathcal{X}$, a mapping $\mathbf{H}:~\mathcal{X} \rightarrow \mathcal{X}$ is a contraction mapping, or  contraction, if there exists a constant $\delta$, with $delta \in (0,1)$, such that  
\begin{eqnarray}
\begin{aligned}
\|\mathbf{H}x_1- \mathbf{H}x_2 \| \leq \delta \|x_1 - x_2 \|,
\end{aligned}
\end{eqnarray}
for any $x_1,x_2 \in \mathcal{X}$.
\end{definition}

\begin{proposition}\label{Pro:CMT}
There exists a  contraction mapping $\mathbf{H}$  for the function $q$ with the form of the optimal Q-function in \eqref{Q-func.eq4}. That is 
\begin{eqnarray}\label{CMT.eq}
\begin{aligned}
&\|\mathbf{H}q_1(s_m,\theta_m)- \mathbf{H}q_2(s_m,\theta_m) \|_{\infty}  \leq \delta \|q_1(s_m,\theta_m) - q_2(s_m,\theta_m) \|_{\infty},
\end{aligned}
\end{eqnarray}
\begin{proof}
From \eqref{Q-func.eq3},  the optimal Q-function for {\bf Algorithm \ref{alg:MARL-QL}} can be expressed as 
\begin{eqnarray}\label{Q-func.eq4}
\begin{aligned}
Q_m^*(s_m,\theta_m)& = \sum_{s'_m} F(s_m,\theta_m,s'_m) \times  \big[R(s_m,\theta_m,s'_m) + \delta \max_{\theta'_m} Q_m^*(s'_m,\theta'_m) \big].
\end{aligned}
\end{eqnarray}\noindent
Hence, we have 
\begin{eqnarray}\label{HQ-func.eq}
\begin{aligned}
\mathbf{H}q(s_m,\theta_m) &= \sum_{s'_m} F(s_m,\theta_m,s'_m) \times  \big[R(s_m,\theta_m,s'_m) + \delta \max_{\theta'_m} q(s'_m,\theta'_m) \big].
\end{aligned} 
\end{eqnarray}\noindent
To obtain \eqref{CMT.eq}, we make the following calculations  in \eqref{CMTforQ.eq}. Note that the definition of $q$ is used in  (a), (b) and (c) follows  properties of  absolute value inequalities. Moreover, (d) comes from the definition of  infinity norm  and (e) is based on the maximum calculation. 
\newcounter{mytempeqncnt}
\begin{figure*}[ht]
\normalsize
\setcounter{mytempeqncnt}{\value{equation}}
\setcounter{equation}{9} 
{\small
\begin{eqnarray}\label{CMTforQ.eq}
\begin{aligned}
&\|\mathbf{H}q_1(s_m,\theta_m)- \mathbf{H}q_2(s_m,\theta_m) \|_{\infty} 
\stackrel{\mathrm{(a)}}{=} \max_{s_m,\theta_m} \delta \bigg|\sum_{s'_m} F(s_m,\theta_m,s'_m)  \big[ \max_{\theta'_m} q_1(s'_m,\theta'_m) -  \max_{\theta'_m} q_2(s'_m,\theta'_m)\big]\bigg|\\
&\stackrel{\mathrm{(b)}}{\leq} \max_{s_m,\theta_m} \delta \sum_{s'_m} F(s_m,\theta_m,s'_m)  \bigg| \max_{\theta'_m} q_1(s'_m,\theta'_m) -  \max_{\theta'_m} q_2(s'_m,\theta'_m)\bigg| \\
&\stackrel{\mathrm{(c)}}{\leq}  \max_{s_m,\theta_m} \delta \sum_{s'_m} F(s_m,\theta,s'_m) \max_{\theta'_m} \bigg| q_1(s'_m,\theta'_m) -   q_2(s'_m,\theta'_m)\bigg| \\
&\stackrel{\mathrm{(d)}}{=} \max_{s_m,\theta_m} \delta  \sum_{s'_m} F(s_m,\theta,s'_m)  \| q_1(s'_m,\theta'_m) -   q_2(s'_m,\theta'_m)\|_{\infty}  \stackrel{\mathrm{(e)}}{=} \delta   \| q_1(s'_m,\theta'_m) -   q_2(s'_m,\theta'_m)\|_{\infty} 
\end{aligned}
\end{eqnarray}}
\vspace*{-10pt} 
\end{figure*}
\setcounter{equation}{10}

\end{proof}
\end{proposition}

Based on \eqref{Psi.eq} and \eqref{HQ-func.eq}, 
\begin{eqnarray}\label{Psi-func.eq1}
\begin{aligned}
E\{\Psi^{t}(s_m,\theta_m)\} &= \sum_{s'_m} F(s_m,\theta,s'_m) \times  \big[r_m^t + \delta \max_{\theta'_m \in \Theta_m} Q_m^{t}(s'_m,\theta'_{m})  - Q_m^*(s_m,\theta_m) \big] \\
& = \mathbf{H}Q^t_m(s_m,\theta_m) -  Q_m^*(s_m,\theta_m)\\
& = \mathbf{H}Q^t_m(s_m,\theta_m) - \mathbf{H} Q_m^*(s_m,\theta_m).
\end{aligned} 
\end{eqnarray}\noindent
where we have used the fact that $Q_m^*(s_m,\theta_m) = \mathbf{H}Q^*_m(s_m,\theta_m) $ since  $Q_m^*(s_m,\theta_m)$  is a some constant value. As a result, we can obtain from {\bf Proposition \ref{Pro:CMT}} and \eqref{Triangle.eq} that
\begin{eqnarray}\label{Psi-func.eq2}
\begin{aligned}
\|E\{\Psi^{t}(s_m,\theta_m)\}\|_{\infty} &\leq \delta \|Q^t_m(s_m,\theta_m) -  Q_m^*(s_m,\theta_m) \|_{\infty} \\&= \delta \| \bigtriangleup_m^t(s_m,\theta_m) \|_{\infty},
\end{aligned}
\end{eqnarray}
Note that \eqref{Psi-func.eq2} corresponds to condition 3) of {\bf Lemma \ref{Lemma:StoApproxi}} in the form of  infinity norm. 

Finally, we verify  the condition in 4) of {\bf Lemma \ref{Lemma:StoApproxi}} is satisfied.
\begin{eqnarray}
\begin{aligned}
&\mathrm{Var}\{\Psi^{t}(s_m,\theta_m) \} \\
&= E\{ r_m^t + \delta \max_{\theta'_m \in \Theta_m} Q_m^{t}(s'_m,\theta'_{m})  - Q_m^*(s_m,\theta_m) -  \mathbf{H}Q^t_m(s_m,\theta_m) +  Q_m^*(s_m,\theta_m) \}\\
& = E\{ r_m^t + \delta \max_{\theta'_m \in \Theta_m} Q_m^{t}(s'_m,\theta'_{m})  -  \mathbf{H}Q^t_m(s_m,\theta_m)  \}\\
& = \mathrm{Var}\{ r_m^t + \delta \max_{\theta'_m \in \Theta_m} Q_m^{t}(s'_m,\theta'_{m})  \}\\
&\leq C(1+ \| \bigtriangleup_m^t(s_m,\theta_m)\|_W^2),
\end{aligned}
\end{eqnarray}
where $C$ is some constant. The final step is based on the fact that the variance of $r_m^t$ is bounded and $Q_m^{t}(s'_m,\theta'_{m})$ at most linearly.

Therefore,  $\|\bigtriangleup_m^t(s_m,\theta_m)\|$ converges to zero w.p.1 in {Lemma \ref{Lemma:StoApproxi}}, which indicates $Q_m^{t}(s_m,\theta_{m})$ converges to $Q_m^*(s_m,\theta_{m})$ w.p.1 in {\bf Theorem \ref{Theo:ConvergenceQ}}.

\vspace{-0.3cm}
{\small
 \bibliographystyle{IEEEtran}
  \linespread{1.1}\selectfont
\bibliography{myref}
}
\end{document}